\documentclass[amsmath,onecolumn, floatfix,nofootinbib, prx, aps, usenames, dvipsnames,11pt]{quantumarticle}
\pdfoutput=1

\usepackage[utf8]{inputenc}
\usepackage[english]{babel}
\usepackage[T1]{fontenc}
\usepackage{natbib}
\usepackage{bbold}

\usepackage[page]{totalcount}

\usepackage{qcircuit}

\usepackage[usenames,dvipsnames]{xcolor}
\usepackage{graphicx}
\usepackage[ colorlinks = true, 
             linkcolor = MidnightBlue,
             urlcolor  = MidnightBlue,
             citecolor = orange,
             anchorcolor = ForestGreen,
]{hyperref} 
\usepackage{float}

\usepackage[sans]{dsfont} 
\usepackage{bbm}
\usepackage[mathscr]{euscript}
\usepackage{cancel} 

\usepackage{chngcntr}

\usepackage{tcolorbox} 
\usepackage{framed}
\usepackage{mdframed}

\usepackage{amsmath, amssymb, amsfonts, amsthm}
\usepackage{units}
\usepackage{mathtools}
\usepackage{thmtools, thm-restate} 

\usepackage{ucs}
\usepackage{subcaption}
\usepackage{fullpage}
\usepackage{csquotes}
\usepackage{epigraph}

\usepackage{cleveref}

\usepackage{enumerate}

\tcbuselibrary{theorems}

\tcbset{colback=orange!5,colframe=orange!75!yellow, 
fonttitle= \bfseries, float=htb}

\newcommand{\footnoterecall}[1]{\hyperref[#1]{\footnotemark[\value{#1}]}}

\declaretheoremstyle[
    spaceabove=6pt, 
    spacebelow=6pt, 
    headfont=\normalfont\bfseries,
    notefont=\mdseries\bfseries, 
    notebraces={(}{)}, 
    bodyfont=\normalfont\itshape,
    postheadspace=1em,
    headpunct={:}]{thmstyle}

\newtheorem{theorem}{Theorem}[section]

\newtheorem{lemma}[theorem]{Lemma}

\newtheorem{definition}[theorem]{Definition}

\numberwithin{equation}{section} 
\counterwithin{figure}{section}
    \newcommand{\ket}[1]{\vert  #1 \rangle}
    \newcommand{\bra}[1]{\langle #1 |}
    \newcommand{\inprod}[2]{\langle #1 | #2 \rangle}
	\newcommand{\proj}[2]{\ket{#1}\bra{#2}}

	\newcommand{\hilbert}{\mathcal{H}}

	\newcommand{\tr}{\operatorname{Tr}  }
	
	\renewcommand{\vec}[1]{\mathbf{#1}}

\usepackage{tikz}
\usetikzlibrary{arrows}
\usepackage{rotating}
\usetikzlibrary{shapes}
\usetikzlibrary{hobby}
\usetikzlibrary{decorations.markings}

\usepackage{multirow}
\usepackage{collcell}
\usepackage{multicol} 

\usetikzlibrary{tikzmark}

\newcommand{\ApplyGradient}[1]{%
        \pgfmathsetmacro{\PercentColor}{100}
        \hspace{-0.33em}\colorbox{cyan!\PercentColor!white}{}
}

\newcommand{\colorcell}[1]{%
        \hspace{0em}\colorbox{#1}{}
}

\newcolumntype{X}{>{\collectcell\colorcell}c<{\endcollectcell}}

\newcolumntype{R}{>{\collectcell\ApplyGradient}c<{\endcollectcell}}

\newcommand{\SetNumber}[1]{%
        \hspace{0.50em}#1\hspace{0.50em}
}

\newcolumntype{C}{>{\collectcell\SetNumber}c<{\endcollectcell}}
\newcolumntype{M}{>{\collectcell\SetNumber}c<{\endcollectcell}}
\renewcommand{\arraystretch}{0}
\begin{document}

\title{A consolidating review of Spekkens' toy theory}

\author{Ladina Hausmann}
\affiliation{Institute for Theoretical Physics, ETH Zurich, 8093 Z\"{u}rich, Switzerland}
\email{l@dinahausmann.ch}

\author{Nuriya Nurgalieva}
\affiliation{Institute for Theoretical Physics, ETH Zurich, 8093 Z\"{u}rich, Switzerland}
\email{nuriya@phys.ethz.ch}

\author{Lídia del Rio}
\affiliation{Institute for Theoretical Physics, ETH Zurich, 8093 Z\"{u}rich, Switzerland}
\email{lidia@phys.ethz.ch}

\date{}

\begin{abstract}
In order to better understand a complex theory like quantum mechanics,  it is sometimes useful to take a step back and create alternative theories, with more intuitive foundations, and examine which features of quantum mechanics can be reproduced by such a \emph{foil theory}.
A prominent example is Spekkens' toy theory, which is based off a simple premise: ``What if we took a common classical theory and added the uncertainty principle as a postulate?'' 
In other words, the theory imposes an \emph{epistemic restriction} on our knowledge about a physical system: only half of the variables can ever be known to an observer. 
Like good science fiction, from this simple principle a rich behaviour emerges, most notoriously when we compose several systems. The toy theory emulates some aspects of  quantum non-locality, although crucially it is still a non-contextual model. 
In this pedagogical review we consolidate  different approaches to Spekkens' toy theory, including the stabilizer formalism and the generalization to arbitrary dimensions, completing them with new results where necessary. In particular, we introduce a general  characterization of measurements, superpositions and entanglement in the toy theory. 
\end{abstract}

\maketitle

\setlength{\epigraphwidth}{5in}
\epigraph{``Fiction writers, at least in their braver moments, do desire the truth: to know it, speak it, serve it. But they go about it in a peculiar and devious way, which consists in inventing persons, places, and events which never did and never will exist or occur, and telling about these fictions in detail and at length and with a great deal of emotion, and then when they are done writing down this pack of lies, they say, There! That’s the truth.''}{Ursula K.~Le Guin, \emph{The Left Hand of Darkness}}

\newpage

\subsection*{\LARGE What to expect from this review}
\label{sec:introduction}
\paragraph{Goals and audience.}
This article gives a pedagogical overview of different approaches to Spekkens' toy theory. We hope it is useful for readers  new to the theory, who wish to get a smooth introduction, and also to those familiar with a few results, who are looking for a consolidated picture or a quick resource for consultation. 

\paragraph{Approaches covered.}
We present three formulations of Spekkens' toy theory: the original formulation by Spekkens~\cite{Spekkens_2007}, its reformulation in the stabilizer notation by Pusey~\cite{Pusey_2012}, and a generalization to systems of an arbitrary dimension by Spekkens~\cite{SpekkensFoundations2016}. Even though the restrictions on knowledge appear different for the different versions, the first two are equivalent and are a special case of the generalized toy theory~\cite{SpekkensFoundations2016,Pusey_2012, Catani_2017}. In short, the original formulation is where we get the intuition, the stabilizer approach is better suited to prove results and find quantum analogs, and the generalization takes the best of both. 
For each formulation, we define epistemic states, measurements, mixtures, superpositions, entanglement and  allowed transformations. 

\paragraph{New results.}
In addition to reviewing current approaches, this article adds a number of contributions and simplifications to Spekkens' toy theory, mostly concerning the generalization to arbitrary dimensions. These are listed in the conclusions (section~\ref{sec:discussion:results}). 
Our new results are blended into the natural section referring to each topic, among established results (and not in an independent section ``results''). Existing results in the literature have explicit references; new lemmas, theorems and definitions are those without a reference. 
All proofs can be found in the appendix.

\tableofcontents

\section{Epistemic states}
\label{sec:states}
\begin{figure}[t]
\centering
\includegraphics[scale=0.2]{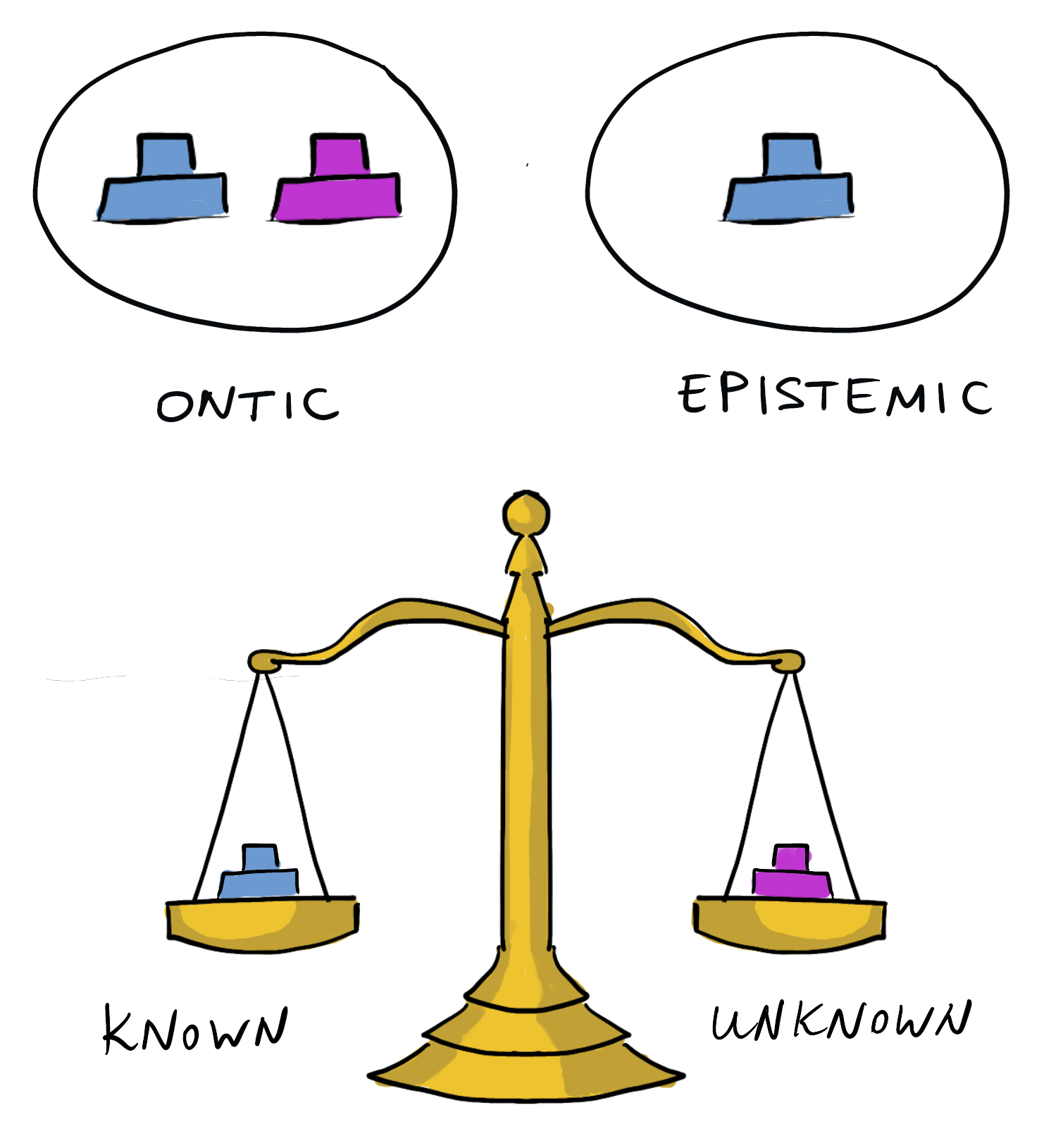}
\caption{{\bf Knowledge balance principle~\cite{Spekkens_2007}.} The toy theory imposes that the amount of maximal knowledge one has about the ontic state (one's \textit{epistemic} state) is equal to the amount of knowledge one lacks. For example, if the ontic state can be fully characterized by a set of two questions, the epistemic state would contain answers to only one of them.}
\label{fig:balance}
\end{figure}

Spekkens' toy theory is an \textbf{epistemically-restricted theory}~\cite{Spekkens_2005}. Such theories distinguish two types of states: \textbf{ontic states}, which encode the physical state of a system, and \textbf{epistemic states} -- the states of knowledge that an observer has about the system. 
In the toy theory, epistemic states  are constrained by the 
 \emph{knowledge balance principle} (Figure~\ref{fig:balance}), inspired by the Heisenberg uncertainty principle:
            \begin{displayquote}
                ``If one has maximal knowledge, then for every system,
                at every time, the amount of knowledge one possesses
                about the ontic state of the system at that time must
                equal the amount of knowledge one lacks.'' \cite{Spekkens_2007}
            \end{displayquote}
We will see how this principle is implemented in different formalizations of the theory, and how it restricts the set of allowed states and the behaviour of the toy theory.

\subsection{Original toy theory: epistemically-restricted picture}

The knowledge balance principle is a restriction on the amount of knowledge one can have, and to understand this principle we first need to quantify knowledge. As a quantifier, Spekkens takes the amount of questions in a canonical set that an observer can answer. A \textbf{canonical set} is defined as the minimal set of questions that  can fully specify each ontic state with different combinations of answers. For example, assume that we only allow for binary questions,\footnote{In principle we can build theories where we would use canonical questions with three or more answers. However, these theories cannot easily be connected to each other \cite{Spekkens_2007}.} and consider a system with four different ontic states, labelled 1, 2, 3 and 4. As this is a two-bit system, a canonical set has two questions, for example
            \begin{displayquote}
                \{``Is the ontic state in $\{1,2 \} $?'', ``is the ontic state in $\{1,3 \} $?''\}.
            \end{displayquote}
For instance, an answer of (no, no) would single out the ontic state as 4.
The knowledge-balance principle restricts an observer to only know the answer to one of these questions. Suppose that we know the answer to the first question is ``yes''; then the epistemic state would be ``the ontic state is in $\{1,2\}$''. We denote this state as $1 \lor 2$. 

\paragraph{Graphical representation.} Epistemic states can also be graphically represented; in this representation we assign each ontic state as a cell in a grid,

            \renewcommand{\arraystretch}{1}
            \begin{equation}
                \begin{tabular}{C C C C}

                    \hline
                    \multicolumn{1}{|c}{ }&
                    \multicolumn{1}{|c}{ }&\multicolumn{1}{|c}{ }&
                    \multicolumn{1}{|c|}{ }\\
                    \hline
                    1 & 2 & 3& 4\\

                \end{tabular}
                \quad.
            \end{equation}
The epistemic state $1 \lor 2$ is then represented by marking the cells where the ontic state could be,
            \renewcommand{\arraystretch}{0}

            \begin{equation}
                \begin{tabular}{|{c}@{}|@{}{c}|{c}|{c}|}
                    \hline
                
                \colorcell{cyan}&\colorcell{cyan}&\colorcell{white}&\colorcell{white}\\ \hline
               
                \end{tabular}
                \quad .
              \end{equation}
            
\paragraph{Further restrictions.} To fully describe the toy theory, we need three more assumptions \cite{Spekkens_2007}. 
 Firstly, we only consider physical systems
            that allow for a knowledge-balance principle. In particular, each ontic state must be uniquely described by a set
            of answers to a canonical set of questions, and the number
            of questions needs to be divisible by two. If
            there are only two questions in the canonical set, we say
            that such a system is \textbf{elementary}, like in the above example.
            Secondly, we
            assume that every system is composed 
            of elementary systems,
            and that the knowledge balance principle applies to each possible set of
            subsystems. We will see that this assumption additionally
            constrains the restricts epistemic states.
            Thirdly, we assume that the epistemic
            state does not constrain the evolution of a system, as it
            seems implausible that the state of mind of an observer
            could influence the system.\footnote{This assumption will be addressed in future work.} 
            
\paragraph{Pure and mixed states.}  We call an epistemic state a maximal
            information state or \emph{pure} if exactly half of the answers
            to the questions in a canonical set are known; otherwise we
            call such a state non-maximal informational or \emph{mixed}.

           \paragraph{Composed systems.}  If a system is composed of $N$ elementary
            subsystems, it has $2N$ binary questions in the canonical
            set (two for each subsystem), and therefore has $2^{2N}$ ontic states. Note that  any
            set of questions that uniquely characterizes any ontic
            state with $2N$ elements is a canonical set of questions,
            because differentiating $2^{2N}$ ontic states requires $2N$
            bits of information. We denote an ontic state $o$ on
            $N$ subsystems as the composition of the ontic states of each single systems $o = (o_1 \times o_2 \times
            ...\times o_N)$, where $o_i$ is the ontic state of the $i$th subsystem. The symbol $\times$ represents the composition of the toy systems.
            
            

            We can extend the graphical
            notation for systems composed of two elementary subsystems.
            Each ontic state is represented by a field in a $4 \times
            4$ grid, such that the ontic state $(o_A \times o_B)$ is the
            in the $o_A$th row and the $o_B$th column as in the grid below:
            
            \renewcommand{\arraystretch}{1}
            \begin{equation}
                \begin{tabular}{c C C C C C}
                    \cline{3-6}
                    \multirow{4}{*}{A \hspace{0.05cm} }& \multicolumn{1}{c|}{4 \hspace{0.1cm}}&
                    \multicolumn{1}{c|}{ }&
                    \multicolumn{1}{c|}{ }&\multicolumn{1}{c|}{ }&
                    \multicolumn{1}{c|}{ }\\
                    \cline{3-6}
                    &\multicolumn{1}{c|}{ 3 \hspace{0.1cm}}&
                    \multicolumn{1}{c|}{ }&
                    \multicolumn{1}{c|}{ }&\multicolumn{1}{c|}{ }&
                    \multicolumn{1}{c|}{ }\\
                    \cline{3-6}
                    &\multicolumn{1}{c|}{ 2 \hspace{0.1cm}}&
                    \multicolumn{1}{c|}{ }&
                    \multicolumn{1}{c|}{ }&\multicolumn{1}{c|}{ }&
                    \multicolumn{1}{c|}{ }\\
                    \cline{3-6}
                    &\multicolumn{1}{c|}{ 1 \hspace{0.1cm}}&
                    \multicolumn{1}{c|}{ }&
                    \multicolumn{1}{c|}{ }&\multicolumn{1}{c|}{ }&
                    \multicolumn{1}{c|}{ }\\
                    \cline{3-6}
                    &\multicolumn{1}{c}{ }&1 & 2 & 3& 4\\ 
                    \multicolumn{1}{c}{}&\multicolumn{1}{c}{}&\multicolumn{4}{c}{B}

                \end{tabular} \quad.
            \end{equation}
            \renewcommand{\arraystretch}{0}
            An epistemic state is then represented by the fields where
            the ontic state could be. For example, the state $(1 \times
            1) \lor (1 \times
            2) \lor (2 \times
            1) \lor (2 \times
            2) $ would be represented by

            \begin{equation}
                \begin{tabular}{|{c}@{}|@{}{c}|{c}|{c}|}
                    \hline
                \colorcell{white}&\colorcell{white}&\colorcell{white}&\colorcell{white}\\ \hline
                \colorcell{white}&\colorcell{white}&\colorcell{white}&\colorcell{white}\\ \hline
                \colorcell{cyan}&\colorcell{cyan}&\colorcell{white}&\colorcell{white}\\ \hline
                \colorcell{cyan}&\colorcell{cyan}&\colorcell{white}&\colorcell{white}\\ \hline
               
                \end{tabular}\quad.
              \end{equation}
            For systems composed of more than two subsystems, states
            can be represented in a $N$-dimensional grid where each
            side has length $4$ \cite{Spekkens_2007}. 
            
            Let us investigate which epistemic states are allowed
            under these conditions. To do so, we need to characterize the canonical sets of questions, and find out what the knowledge balance principle entails for epistemic states. We show that it breaks permutation symmetry and imposes an inductive structure.

            To formalize the notion of  \emph{question} in the toy theory, we introduce a map $\mathcal{M}$ that maps a question
            \begin{equation}
                Q = \text{``Is the
            ontic state in $M_0$?''}
            \end{equation}
            to a partition of all ontic states $O$:
            \begin{equation}
                \mathcal{M}: Q \to (M_0, O \setminus M_0) =
            (M_0, M_1).
            \end{equation}
            We can write the answer to a question $Q_i$ as a bit $a_i$: 0 for ``yes, the ontic state is in $M_{i,0}$'' and 1 for ``no, the ontic state is not in $M_{i,0}$; it is in $M_{i,1}$ instead.'' 
            For a  given set of $k$ questions $\mathcal Q$, we can represent a hypothetical answer as a $k$-bit string $\vec a$ where each individual bit $a_i \in \{0,1\}$ is the answer to question $Q_i$.
            Defining these partitions allows us to  characterize canonical sets (proofs in the appendix). 

            \begin{restatable}[Canonical set of questions]{lemma}{toyorgcanonical}\label{toy:org:canonical}
                Let $\mathcal{Q}$ be a set of $2N$ questions for a system
                with $N$ elementary subsystems, $Q_i$ the $i$th
                question and $\mathcal{M}(Q_i) = (M_{i,0},M_{i,1})$.
                
                Then $\mathcal{Q}$ is a canonical set of questions if
                and only if for every $2N$-bit string $\vec a$ the
                intersection $\cap_{i} M_{i,a_i}$ contains exactly one
                element. 
            \end{restatable}

           We saw that the epistemic state corresponds to a set
           of answers to some questions in a canonical set, which rule out some ontic states.  We call the set of ontic states that are
           allowed under a given epistemic state the \textbf{ontic basis} of the epistemic state \cite{Spekkens_2007}. We will often
           write an epistemic state as $o_1 \lor o_2 \lor ... \lor
           o_{k}$ for $o_1, ..., o_k$ ontic states in its ontic basis. This allows us to calculate the
           ontic basis from the set of answers in
           the following way.
           
           \begin{restatable}[Ontic basis of epistemic states]{lemma}{toyorgbasis} \label{toy:org:basis}
               Let $\mathcal{Q}$ be a canonical set of $2N$ questions on $N$ elementary systems, corresponding to the partitions  $\mathcal M(Q_i) = (M_{i,0}, M_{i,1})$. 
               Let $\mathcal E \subseteq \{1, 2, \dots 2N\} $ index a subset of these questions, and $\mathcal A_{\mathcal E} = \{ a_i\}_{i \in \mathcal E} $  represent answers to this subset of questions (that is, an epistemic state, not necessarily a valid one). 
                Then the ontic basis of this epistemic state is given by the set $E= \cap_{i \in \mathcal E} M_{i, a_i}$. 
           \end{restatable}

           Lemmas~\ref{toy:org:canonical} and~\ref{toy:org:basis} allow us to find a necessary
           condition for the validity of epistemic states. 

           \begin{restatable}[Dimension of valid epistemic states]{theorem}{toyorgnumber}\label{toy:org:number}
               Let $E$ be an ontic basis of a valid epistemic state
               on $N$ elementary systems, which covers the answers to a subset of $k$ out of $2N$ canonical questions. Then $|E| =
               2^{2N-k}$.
           \end{restatable}

        \subparagraph{Simplest example: the toy qubit \cite{Spekkens_2007}.}\label{toy:org:ex:n_1}
                    From the previous discussion, a state on a single elementary
                    system needs to have an ontic basis of size either
                    $2^{2-1} = 2$ or $2^{2-0} = 4$ and any permutation
                    of $2$ or $4$ ontic states in the ontic basis is
                    valid. The valid states are
                    \begin{gather}
                        \begin{tabular}{|{c}@{}|@{}{c}|{c}|{c}|}
                    \hline
                \colorcell{cyan}&\colorcell{cyan}&\colorcell{white}&\colorcell{white}\\ \hline
                \end{tabular}\ ,\qquad
                        \begin{tabular}{|{c}@{}|@{}{c}|{c}|{c}|}
                    \hline
                \colorcell{white}&\colorcell{white}&\colorcell{cyan}&\colorcell{cyan}\\ \hline
                \end{tabular}\  ,\\
                        \begin{tabular}{|{c}@{}|@{}{c}|{c}|{c}|}
                    \hline
                \colorcell{cyan}&\colorcell{white}&\colorcell{cyan}&\colorcell{white}\\ \hline
                \end{tabular}\ ,\qquad
                        \begin{tabular}{|{c}@{}|@{}{c}|{c}|{c}|}
                    \hline
                \colorcell{white}&\colorcell{cyan}&\colorcell{white}&\colorcell{cyan}\\ \hline
                \end{tabular}\ , \\
                        \begin{tabular}{|{c}@{}|@{}{c}|{c}|{c}|}
                    \hline
                \colorcell{cyan}&\colorcell{white}&\colorcell{white}&\colorcell{cyan}\\ \hline
                \end{tabular}\ ,\qquad
                        \begin{tabular}{|{c}@{}|@{}{c}|{c}|{c}|}
                    \hline
                \colorcell{white}&\colorcell{cyan}&\colorcell{cyan}&\colorcell{white}\\ \hline
                \end{tabular}\ ,
                    \end{gather}
                    and the non-maximal information state
                    \begin{equation}\label{toy:org:fullymixed}
                        \begin{tabular}{|{c}@{}|@{}{c}|{c}|{c}|}
                    \hline
                \colorcell{cyan}&\colorcell{cyan}&\colorcell{cyan}&\colorcell{cyan}\\ \hline
                \end{tabular}\ .
                    \end{equation}

                    This single system is analogous to a qubit system
                    quantum mechanics. In particular, we can identify
                    the six maximum information states with the following pure states of a qubit \cite{Spekkens_2007}:

                    \begin{align}
                        \begin{split}
                            (1\lor 2) &\iff |0\rangle, \\
                            (3\lor 4) &\iff |1\rangle ,\\
                            (1\lor 3) &\iff |+\rangle ,\\
                            (2\lor 4) &\iff |-\rangle ,\\
                            (2\lor 3) &\iff |+ i \rangle, \\
                            (1\lor 4) &\iff |- i\rangle ,
                        \end{split}
                    \end{align}
                    where $|\pm\rangle =(|0\rangle \pm
                    |1\rangle)/\sqrt{2}$ and $|\pm i \rangle
                    =(|0\rangle \pm i
                    |1\rangle)/\sqrt{2}$,
                    and the non-maximal information state with the
                    fully mixed qubit state \cite{Spekkens_2007}:

                    \begin{equation}
                        (1\lor 2 \lor 3 \lor 4) \iff \frac{\mathbbm 1}2 .
                    \end{equation}

           \paragraph{Valid states at the level of subsystems.} Theorem~\ref{toy:org:number} only says something about the
           knowledge balance principle applied to the global state.
           However, for a state to be valid, the knowledge balance
           principle also has to be fulfilled on each subsystem. As an example consider the valid bipartite state
           $(1\cdot 1)\lor (1\cdot 2)\lor  (2\cdot 1)\lor (2\cdot 2)$
           graphically represented as

           \begin{equation}
            \begin{tabular}{|{c}@{}|@{}{c}|{c}|{c}|}
                    \hline
                \colorcell{white}&\colorcell{white}&\colorcell{white}&\colorcell{white}\\ \hline
                \colorcell{white}&\colorcell{white}&\colorcell{white}&\colorcell{white}\\ \hline
                \colorcell{cyan}&\colorcell{cyan}&\colorcell{white}&\colorcell{white}\\ \hline
                \colorcell{cyan}&\colorcell{cyan}&\colorcell{white}&\colorcell{white}\\ \hline
                \end{tabular} \quad.
           \end{equation}
            The knowledge balance principle applied to the entire
            state is fulfilled for any permutation of the ontic
            states. We can apply the permutation to the partition
            corresponding to the questions in the canonical set and
            get a new permuted canonical set of questions, because the
            condition in \cref{toy:org:canonical} is still fulfilled.
            Thus, the permuted epistemic state fulfills the knowledge
            balance principle at the global. However, some of these permutations, like the state
            $(1 \cdot 1) \lor (2 \cdot  1) \lor (3 \cdot  1) \lor (4
            \cdot  1)$,

                \begin{equation}
                    \begin{tabular}{|{c}@{}|@{}{c}|{c}|{c}|}
                    \hline
                \colorcell{cyan}&\colorcell{white}&\colorcell{white}&\colorcell{white}\\ \hline
                \colorcell{cyan}&\colorcell{white}&\colorcell{white}&\colorcell{white}\\ \hline
                \colorcell{cyan}&\colorcell{white}&\colorcell{white}&\colorcell{white}\\ \hline
                \colorcell{cyan}&\colorcell{white}&\colorcell{white}&\colorcell{white}\\ \hline
                \end{tabular} \quad , 
                \label{eq:invalid_local_states}
                \end{equation}
            are not valid: in this example, the
            knowledge balance principle is not fulfilled on the second
            system, where the local ontic state is known to be $1$.
            This example shows that the knowledge balance principle on
            the entire state is permutation-symmetric, but when we
            require it to hold for every
            subsystem  we break this permutation symmetry.

            \subparagraph{Local permutations.} Nevertheless, some permutations still map valid states to valid states, even for the composed systems.  These are \textbf{local} permutations: those that act on individual elementary subsystems, leaving the others untouched. This fact leads to an equivalence relation between states \cite{Spekkens_2007}.

                    \begin{definition}[Equivalent epistemic states] A permutation of a composed system is called local if it acts only on an individual elementary subsystem. 
                        Two epistemic states are equivalent if and
                        only if there exists a local permutation that
                        maps one to other.
                    \end{definition}

                    Spekkens shows that this is indeed an equivalence relation \cite{Spekkens_2007}.
                    
            \subparagraph{Example: two elementary systems \cite{Spekkens_2007}.}
                For a system composed of two elementary subsystems (the analogous of two qubits), the equivalence classes of states are those of the pure states
                    \begin{align*}
                        \begin{tabular}{|{c}@{}|@{}{c}|{c}|{c}|}
                    \hline
                \colorcell{white}&\colorcell{white}&\colorcell{white}&\colorcell{white}\\ \hline
                \colorcell{white}&\colorcell{white}&\colorcell{white}&\colorcell{white}\\ \hline
                \colorcell{cyan}&\colorcell{cyan}&\colorcell{white}&\colorcell{white}\\ \hline
                \colorcell{cyan}&\colorcell{cyan}&\colorcell{white}&\colorcell{white}\\ \hline
                \end{tabular} \iff  \ket 0  \otimes \ket 0, \quad 
                        \begin{tabular}{|{c}@{}|@{}{c}|{c}|{c}|}
                    \hline
                \colorcell{white}&\colorcell{white}&\colorcell{white}&\colorcell{cyan}\\ \hline
                \colorcell{white}&\colorcell{white}&\colorcell{cyan}&\colorcell{white}\\ \hline
                \colorcell{white}&\colorcell{cyan}&\colorcell{white}&\colorcell{white}\\ \hline
                \colorcell{cyan}&\colorcell{white}&\colorcell{white}&\colorcell{white}\\ \hline
                \end{tabular} \iff \frac{1}{\sqrt{2}} (|0\rangle \otimes |0\rangle + |1\rangle \otimes |1\rangle),
                    \end{align*}
                and the mixed states 
                    \begin{align*}
                        \begin{tabular}{|{c}@{}|@{}{c}|{c}|{c}|}
                    \hline
                \colorcell{white}&\colorcell{white}&\colorcell{white}&\colorcell{white}\\ \hline
                \colorcell{white}&\colorcell{white}&\colorcell{white}&\colorcell{white}\\ \hline
                \colorcell{cyan}&\colorcell{cyan}&\colorcell{cyan}&\colorcell{cyan}\\ \hline
                \colorcell{cyan}&\colorcell{cyan}&\colorcell{cyan}&\colorcell{cyan}\\ \hline
                \end{tabular} \iff \proj00  \otimes \frac{\mathbb{1}}2, \quad
                        \begin{tabular}{|{c}@{}|@{}{c}|{c}|{c}|}
                    \hline
                \colorcell{cyan}&\colorcell{cyan}&\colorcell{white}&\colorcell{white}\\ \hline
                \colorcell{cyan}&\colorcell{cyan}&\colorcell{white}&\colorcell{white}\\ \hline
                \colorcell{cyan}&\colorcell{cyan}&\colorcell{white}&\colorcell{white}\\ \hline
                \colorcell{cyan}&\colorcell{cyan}&\colorcell{white}&\colorcell{white}\\ \hline
                \end{tabular} \iff \frac{\mathbb{1}}2  \otimes \proj00 ,
                    \end{align*}
                    \begin{align*}
                        \begin{tabular}{|{c}@{}|@{}{c}|{c}|{c}|}
                    \hline
                \colorcell{white}&\colorcell{white}&\colorcell{cyan}&\colorcell{cyan}\\ \hline
                \colorcell{white}&\colorcell{white}&\colorcell{cyan}&\colorcell{cyan}\\ \hline
                \colorcell{cyan}&\colorcell{cyan}&\colorcell{white}&\colorcell{white}\\ \hline
                \colorcell{cyan}&\colorcell{cyan}&\colorcell{white}&\colorcell{white}\\ \hline
                \end{tabular} \iff \frac{1}{2}(\proj00  \otimes \proj00 + \proj11  \otimes \proj11), 
                    \end{align*}
                    \vspace{0.5cm}
                    and
                    \begin{align*}
                        \begin{tabular}{|{c}@{}|@{}{c}|{c}|{c}|}
                    \hline
                \colorcell{cyan}&\colorcell{cyan}&\colorcell{cyan}&\colorcell{cyan}\\ \hline
                \colorcell{cyan}&\colorcell{cyan}&\colorcell{cyan}&\colorcell{cyan}\\ \hline
                \colorcell{cyan}&\colorcell{cyan}&\colorcell{cyan}&\colorcell{cyan}\\ \hline
                \colorcell{cyan}&\colorcell{cyan}&\colorcell{cyan}&\colorcell{cyan}\\ \hline
                \end{tabular} \iff \frac{\mathbb{1}}4. 
                    \end{align*}

            \subparagraph{Inductive validity checks.} The requirement that the knowledge balance principle has
            to apply to each subsystem leads to an inductive
            structure: let 
            \begin{equation}
                E_{1,2,\dots N} = (a_{1,1} \cdot ... \cdot a_{1,N}) \lor (a_{2,1} \cdot ... \cdot a_{2,N}) \lor ... \lor (a_{2^{k},1} \cdot ... \cdot a_{2^{k},N})
            \end{equation}
            be a state whose validity needs to be checked. First we need to check that the state is globally valid. Then the
            knowledge balance requires that each reduced state of $N-1$ subsystems be
            a valid epistemic state too. For example the marginal over
            system 1 (equivalent to tracing out system 1)
            \begin{equation}
                E_{2,\dots N} = (a_{1,2} \cdot ... \cdot a_{1,N}) \lor (a_{2,2} \cdot ... \cdot a_{2,N}) \lor ... \lor (a_{2^{k},2} \cdot ... \cdot a_{2^{k},N})
            \end{equation}
            needs to be a valid state. So we check whether this reduced state is valid, first at a global level\dots But for this to be a valid state, each of its $(N-2)$-subsystem marginals must be valid, and so on. 
            Thus, the knowledge balance principle leads to an inductive
            structure, and a lengthy process to verify validity of composed states. 
            
            Note that it would not be sufficient to simply check the reduced state of all the elementary subsystems, as ``illegal'' correlations between those could make a locally valid state be globally invalid at any of the  levels described above. We go over such an example later on, when we describe measurements; see  (\ref{toy:org:mmt4}).

\subsection{Stabilizer formalism}

In general it
    is hard to directly verify if a given epistemic state over multiple systems is valid, because of the inductive structure of
    epistemic states. However, Pusey \cite{Pusey_2012} noticed that
    the epistemic restriction can be expressed in a way analogous to
    the quantum stabilizer formalism. It is important to note that
    while the quantum stabilizer formalism only allows the description
    of a small corner of the Hilbert space
    \cite{Aaronson_2004,Gottesman_1997}, the toy stabilizer formalism
    can describe every state in the toy theory. The stabilizer approach
    allows us to answer questions that are left open in the previous
    formalism, and to verify the validity of states and operations more efficiently. The proof of the equivalence of the
    stabilizer and epirestricted pictures can be
    found in \cite{Pusey_2012}. 
    
     \paragraph{Quantum stabilizers.}   Let us first briefly review the quantum stabilizer formalism~\cite{Gottesman_1997}; the analogy to the toy stabilizers will be useful ahead. 
    Let $\mathbb{1}_{n}$ be the $n \times n$ identity matrix and 
            \begin{equation}
                X = \begin{pmatrix}
                    0 & 1 \\[6pt]
                    1 & 0 
                    \end{pmatrix},\qquad
                Y = \begin{pmatrix}
                    0 & -i \\[6pt]
                    i & 0 
                    \end{pmatrix},\qquad
                Z = \begin{pmatrix}
                    1 & 0 \\[6pt]
                    0 & -1 
                    \end{pmatrix}
            \end{equation}
            the Pauli matrices. Note that they are not linearly independent as
            $XZ = -iY$. The \emph{Pauli group} on $N$ qubits $P_N$ is defined
            as the following set of matrices
            \begin{equation}
                P_N = \{\alpha M_1 \otimes ...\otimes M_N| \ M_i \in \{\mathbb{1}_{2},X,Y,Z\}, \ \alpha \in \{1,-1,i,-i\}\}.
            \end{equation}
             This group is generated by $\{ i\mathbb{1}, X_k,
             Z_k\}_k$ where $X_k, Z_k$ act on the
            $k$-th qubit with the $X$ or $Z$ operator respectively and trivially on the others.
            
           A \emph{stabilizer group} $S$ is any subgroup of $P_N$
            satisfying the condition that all elements commute and that  $- \mathbb{1}$ is not included. It can be shown that a stabilizer
            subgroup can can have at most $N$ independent elements,  or
            $2^N$ elements in total. We then define the density matrix
            associated with $S$, \footnote{The normalization factor
            of $2^{-N}$ comes from $tr(\rho_s) := 2^{-N}
            \sum_{g\in S} \tr(g) =  2^{-N} tr(\mathbb{1}_{2^{N}})= 1 $,
            as all Pauli operators except the identity have trace zero.}
            \begin{equation}
                \rho_S = 2^{-N} \sum_{g\in S} g.
            \end{equation}
            This is a uniform mixture over $\frac{2^{N}}{|S|}$ states. In
            particular, it is a pure state if $|S| = 2^{N}$. 
            
            Often when the stabilizer formalism is used, one is not directly
            interested in $\rho_S$ but rather in the subspace $V_S$ of states
            $|\psi\rangle$ that are stabilized by $S$, i.e the vectors
            $|\psi\rangle $ such that all $g \in S$ fulfill
            $g|\psi\rangle = |\psi\rangle$. Note that $\rho_S$ is a uniform mixture
            of the basis elements of $V_S$.

       \paragraph{Toy stabilizers.} Now let us relate the formalism above with the language of the original toy theory, following Pusey~\cite{Pusey_2012}. We associate the four elementary ontic states $1, 2, 3, 4$ of the previous
         subsubsection with the vectors $e_1 = (1,0,0,0)^T$, $e_2 =
         (0,1,0,0)^T$, $e_3 = (0,0,1,0)^T$, $e_4 = (0,0,0,1)^T$. We define
         the toy Pauli matrices as
         \begin{equation}
                    \mathcal{X} = \begin{pmatrix}
                            1 & 0 & 0 & 0 \\[6pt]
                            0 & -1 & 0 & 0 \\[6pt]
                            0 & 0 & 1 & 0 \\[6pt]
                            0 & 0 & 0 & -1 
                        \end{pmatrix}, \qquad
                  \mathcal{Y} = \begin{pmatrix}
                            1 & 0 & 0 & 0 \\[6pt]
                            0 & -1 & 0 & 0 \\[6pt]
                            0 & 0 & -1 & 0 \\[6pt]
                            0 & 0 & 0 & 1 
                        \end{pmatrix}, 
            \qquad \mathcal{Z} = \begin{pmatrix}
                            1 & 0 & 0 & 0 \\[6pt]
                            0 & 1 & 0 & 0 \\[6pt]
                            0 & 0 & -1 & 0 \\[6pt]
                            0 & 0 & 0 & -1 
                        \end{pmatrix}.
                \end{equation}
                Similarly to the Pauli matrices, these matrices are not
                independent as $\mathcal{Y} = \mathcal{Z}
                \mathcal{X}$. With these toy Pauli matrices, we define
                the toy Pauli group 
                \begin{equation}
                    G_N = \{\alpha \ p_1 \otimes ... \otimes
                    p_N| \ p_i \in
                    \{\mathbb{1}_{4},\mathcal{X},\mathcal{Y},\mathcal{Z}\}, \ \alpha \in \{1,-1\} \}.
                \end{equation}
                The set $\{\mathcal{X}_k, \mathcal{Z}_k, -\mathbb{1}|
                k \in {1,...,N} \}$ generates $G_N$, where
                $\mathcal{X}_k$ acts as $\mathcal{X}$ on the k-th
                elementary system and $\mathcal{Z}_k$ is defined
                analogously. Each ontic state $o_N = e_{i_1} \otimes
                e_{i_2} ... \otimes e_{i_N}$ is in a different
                combination of eigenspaces of generators of $G_N$.
                Hence, the ontic state is fully characterized by the combination of eigenspaces it lies in.

                \paragraph{``Commutation'' relations.}
                In the quantum stabilizer formalism, commutation relations
                between elements of the Pauli group play a large role.
                In the toy theory though, all elements of $G_N$ commute, as they correspond to diagonal matrices. One can get around this aggravation by imposing ``commutation relations'' by hand, in analogy to the quantum case. That is, we
                define the map $m: G_N \to P_N$ between the toy
                and quantum Pauli groups such that
                $m(\mathcal{X}_k) = X_k$, $m(\mathcal{Z}_k) = Z_k$,
                $m(-\mathbb{1}_{2N}) = - \mathbb{1}_N$ and $m(gh) =
                m(g)m(h)$. We further define a binary relation $\sim$ between two
                elements $g,h \in G_N$  such that $g \sim f$ if and only if
                $m(g)$ and $m(f)$ commute. We forsake rigour of notation for convenience, and  say that in this case $g$ and $f$
                ``commute’’.

                \paragraph{Toy stabilizer groups.}
                In analogy to the quantum stabilizer formalism,
                we can define a stabilizer group $\mathcal{S}$ as a
                commuting subgroup of $G_N$ that does not contain
             $-\mathbb{1}_{2N}$. 
                The epistemic state
                associated with a stabilizer group $\mathcal{S}$ is given by the union of all ontic states stabilized by operators in the group,
                \begin{align*}
                    E_G = \bigvee \{ o_N= e_{i_1} \otimes e_{i_2}\otimes \dots \otimes e_{i_N}: g\ o_N = o_N, \  \forall \ g\in \mathcal S \}. 
                \end{align*}
                For
                example, in an elementary system, the epistemic state $1 \lor 2$ (in  stabilizer
                notation $e_1 \lor e_2$)  corresponds to the
                stabilizer group $\operatorname{span}\{\mathcal{Z}\}$. The
                maximal amount of independent generators of a
                stabilizer group $\mathcal{S}$ is $N$. Epistemic
                states with such a stabilizer group are precisely the pure states.
                
               \paragraph{Relation to original formalism.} We can
                regard each toy Pauli operator as a question where the
                two eigenspaces (for eigenvalues $+1$ and $-1$) can be identified with the
                partition over states corresponding to the question. For example,  $\mathcal{Z}$  corresponds to the question ``Is the ontic state in $\{1,2 \} $?''
                The fact
                that there are at most $N$ independent
                generators in a stabilizer group agrees with the knowledge balance principle, which requires the maximal amount of questions an observer can answer to be $N$, half of the size of the canonical set.
                 The
                commutation requirement ensures that the knowledge
                balance principle is also fulfilled on subsystems.

\subsection{Generalizing to  arbitrary dimensions}

Up to now we only considered the toy theory for systems analogous to qubits. However, the toy theory can be generalized to arbitrary $d$-level and continuous systems. The
generalization of Spekkens' toy theory is inspired by canonical quantization~\cite{Dirac1925},  where a set of observables is  jointly measurable if they all commute, this time relative to the Poisson bracket, instead of the usual matrix commutator~\cite{SpekkensFoundations2016}. 
 
\paragraph{Toy position and momentum.}    
To formalize this notion, we represent a system's $n$ degrees of
        freedom in a language of the ``positions'' $q_1,...,q_n$
        and conjugated ``momenta'' $p_1,...,p_n$.
        The associated phase space is denoted by $\Omega =
        \mathbb{R}^{2n}$ in the continuous case and $\Omega =
        \mathbb{Z}_{d}^{2n}$ in the discrete case.\footnote{ Note that for $d$ not prime $\Omega = \mathbb{Z}_{d}^{2n}$ is
        not a vector space, because $\mathbb{Z}_{d}$ is not a field.
        However, we can see $\mathbb{Z}_{d}^{2n}$ as a module over the
        ring $\mathbb{Z}_{d}$. In this case special care is needed,
        as not every linear algebra result also holds for
        $\mathbb{Z}_{d}^{2n}$ \cite{SpekkensFoundations2016}. In contrary to past research
        \cite{Catani_2017}, in this work we were able to unify and simplify the
        treatment of the prime and non-prime case.} Each point in the phase state $\vec m =
        (q_1,p_1,...,q_n,p_n)^T \in \Omega$ is an \textbf{ontic state} of the
        system. These variables do not have to correspond to actual positions or momenta of particles.

 \paragraph{Quadrature variables.}       Observables  are represented by functionals on the state space, $f:
        \Omega \to \mathbb{R} $ or $\mathbb{Z}_d$.
        Importantly, we will be looking at
        \emph{quadrature variables}: linear combinations of position and
        momentum variables, 
        \begin{align*}
            f(\vec{m}) = a_1\ q_1 + b_1\ p_1 + a_2 \ q_2 + b_2 \ p_2 + \dots + c.
        \end{align*}
        In the discrete case, this expression has to be understood
        within mod $d$. Without loss of generality $c$ can be set to
        zero, because if $f(m)-c$ is known then $f(m)$ is also known.
        Therefore, each quadrature variable can be associated with a
        vector $\vec{f} \in \Omega$ with the coefficients $a_1,b_1,\dots$ and its
        evaluation on an ontic state $\vec m$ can be compactly written as 
        \begin{align*}
            f(\vec{m}) = \vec{f}^T \vec{m}. 
        \end{align*}
        For example, the quadrature variable $q_1 + p_1$ would be
        associated with the vector $\vec{f}^T = (1,1, 0, \dots, 0)$.

\paragraph{Poisson bracket.}        
        The next step is to define the Poisson bracket of two functionals $f$ and $g$. In the continuous case, it follows the definition of classical mechanics \cite{SpekkensFoundations2016},
        \begin{align*}
            [f,g](\vec{m}) = \sum_{i=1}^{n} \left(\frac{\partial f}{\partial \vec q_i}\frac{\partial g}{\partial \vec p_i}-\frac{\partial f}{\partial \vec p_i}\frac{\partial g}{\partial \vec  q_i}\right),
        \end{align*}
        where $\vec{q_i},\vec{p_i}$ are the vectors in phase space
        where all entries are zero except the position or momentum
        of the $i$th degree of freedom, which is $q_i$ or $p_i$ respectively.
        In the discrete case, the
        Poisson bracket can be defined with differences in the
        respective modulo space \cite{SpekkensFoundations2016}:
        \begin{align*}
            \begin{split}
                [f,g](\vec{m}) = \sum_{i=1}^{n} &\Big (f(\vec{m}+\vec{q_i})-f(\vec{m}))(g(\vec{m}+\vec{p_i})-g(\vec{m})\big) \\
                                            &- \big(f(\vec{m}+\vec{p_i})-f(\vec{m}))(g(\vec{m}+\vec{q_i})-g(\vec{m})\Big)_{\text{mod } d} \quad.
            \end{split}
        \end{align*}

\paragraph{Epistemic restriction.}        
     The complete epistemic restriction in then given by
        the principle of classical complementarity:

        \begin{displayquote}
            ``The valid epistemic states are those wherein an agent
            knows the values of a set of quadrature variables that
            commute relative to the Poisson bracket, and is maximally
            ignorant otherwise.'' \cite{SpekkensFoundations2016}
        \end{displayquote}
        Here, ``maximal ignorance'' means that there is a uniform
        probability over all other values of variables.
        It can be shown that this complementarity principle requires observables
        to be linear, i.e.\ quadrature observables.
        
        \paragraph{Sympletic inner product.} 
        Representing quadrature observables as a vector $\vec{f}$ allows for a different expression of the
        commutation rules. We calculate the Poisson bracket for both the
        continuous and discrete case, where the sum has to be understood
        in mod $d$ in the discrete case:
        \begin{align}
            [f,g] = \sum_{i = 1}^n f_{2i - 1} g_{2i} - f_{2i} g_{2i-1}
        \end{align}
        where $f_j$ denotes the $j$th entry of $\vec{f}$. We can rewrite this as a sympletic inner product of the two observables $\vec f $ and $\vec g$, 
        \begin{align}
            [f,g] = \vec{f}^T\ \vec{J}\ \vec{g} =: \langle \vec{f} , \vec{g} \rangle,
        \end{align}
        where we defined the matrix
        \begin{equation}
        \label{eq:j}
            \vec{J} = \begin{pmatrix}
                0 & 1 & 0 & 0 & \dots \\[6pt]
                -1 & 0 & 0 & 0 & \\[6pt]
                0 & 0 & 0 & -1 & \\[6pt]
                0 & 0 & 1 & 0 & \\
                \vdots &  &  & &\ddots  \\
            \end{pmatrix}. 
        \end{equation}
        Therefore two variables commute if and only if they are
        orthogonal with respect to this symplectic inner product \cite{SpekkensFoundations2016}.

        \paragraph{Isotropic subspaces of compatible observables.}
        When an observer knows the value of variables $\vec{f}_1,...,\vec{f}_k$,
     they automatically know the values of all variables in the
        span $\langle \vec{f}_1,...,\vec{f}_k \rangle$. This is because the observer can
        deduce the value of $a \vec{f}_i + b \vec{f}_j$ for known $a,b \in \mathbb{Z}_d
        $ or $\mathbb{R}$ if $\vec{f}_i$ and $\vec{f}_j$ are known. Therefore,
        the known variables form a
        subvector space (in the continuous case or the discrete one with $d$
        prime) or submodule (in the case where $d$ not prime) $V = \langle \vec{f}_1,...,\vec{f}_k
        \rangle$
        of $\Omega$. This subvector space or submodule is called
        \emph{isotropic}, if it satisfies the condition
        \begin{equation}
            \forall \ \vec{f},\vec{g} \in V: \ \langle \vec{f},\vec{g} \rangle = 0.
        \end{equation}
        A vector space or module spanned by commuting elements is
        isotropic by definition. The maximal dimension of an isotropic
        vector space of $\Omega$ is $n$, the number of degrees of
        freedom. Therefore, the maximal amount of variables that can
        be known is $n$, which is exactly half the number of linearly
        independent variables, echoing the knowledge balance
        principle~\cite{SpekkensFoundations2016}.

\paragraph{Valuation vector.}
        For each possible value assignment to a minimal generating set
        $\{\vec{f}_1,...,\vec{f}_k \}$ of an isotropic subvector space or submodule
        corresponds to a vector $\vec{v} \in \Omega$, called the
        valuation vector. For example, if
        $V = \langle \begin{pmatrix} 1 \\[6pt]
            1\end{pmatrix} \rangle$ and we know that $p_1+q_1
            = 0$, the vector $\vec{v}$ is then $\begin{pmatrix} 1 \\[6pt]
                -1\end{pmatrix}$. This vector is not unique: $\vec{v}
                = \begin{pmatrix} 2 \\[6pt]
                    -2 \end{pmatrix}$ would also be an equivalent
                    choice for the valuation vector, as the variables
                    in $V$ have the same valuation for both vectors.

\paragraph{Epistemic states.}        The classical complementarity principle requires that an
        observer be maximally ignorant for all the other variables.
        Therefore, the epistemic state is a probability distribution
        over phase space where we assign equal probability to all ontic
        states $\vec{m} \in \Omega$ which are compatible with the
        knowledge of the observer. More formally, an ontic state
        $\vec{m} \in \Omega$ is compatible with the valuation vector
        $\vec{v}$ if for all $\vec{f} \in V$ 
        \begin{equation}
            \vec{f}^T \vec{m} = \vec{f}^T \vec{v}.
        \end{equation}
        This condition can be rewritten as
        \begin{equation}\label{toy:gen:com_ont}
            \vec{f}^T (\vec{m} -\vec{v}) = 0.
        \end{equation}
        For all ontic states that fulfill the above condition it must
        hold that 
        \begin{equation}
            \vec{m}- \vec{v} \in
        V^{\perp}=\{\vec{m}' \in \Omega| \ \forall \vec{f} \in V \
        \vec{f}^T \vec{m}' = 0\}.
        \end{equation}
        Thus, we can write the set of
        compatible ontic states as the affine subvector space or
        submodule
        
        \begin{equation}
            V^{\perp} + \vec{v} = \{\vec{m}' \in \Omega| \vec{m}'-
        \vec{v} \in V^{\perp}\}.
        \end{equation}
        This forms an equivalence class: the set of ontic states
        that cannot be operationally distinguished through
        observations in $V$, which all correspond to a single epistemic state. We
        summarize the above in the following definition for epistemic
        states. 

        \begin{definition}[Epistemic state (generalized)] Let $V= \langle
        \vec{f}_1,...,\vec{f}_k \rangle$ be the isotropic space/module
        generated by commuting observables $\vec{f}_1,..,\vec{f}_k$,
        that is a set of mutually knowable observables. 
        Let $\vec{v}$ be one possible valuation vector for these
        observables. Together, $\vec{v}$ and $V$ form an epistemic
        state $(V,\vec{v})$. The set of ontic states compatible with that
        valuation is given by $V^{\perp} + \vec{v}$.
        Intuitively, this corresponds to the set of possible ontic
        states, given that we measured the observables in $V$ and
        obtained valuation $\vec{v}$. 
        \end{definition} 

        The resulting probability distribution of possible ontic
        states for a given epistemic state, assuming classical
        complementarity over phase space
        $\mu_{V,\vec{v}}$ is then given by
        \begin{equation}\label{eq:gen:prob_state}
            \mu_{V,\vec{v}}(\vec{m}) = \frac{1}{N_{V}} \delta_{V^{\perp} + \vec{v}}(\vec{m})
        \end{equation}
        with 

        \begin{equation}
            \delta_{V^{\perp} + \vec{v}}(m) = \prod_{\vec{f} \in V} \delta(\vec{f}^T \vec{m}- \vec{f}^T \vec{v})
        \end{equation}
        and $\frac{1}{N_{V}}$ a normalization constant. In the
        discrete case $\delta$ is the Kronecker delta and in the
        continuous case it is the Dirac delta function \cite{SpekkensFoundations2016}.

\paragraph{Pure and mixed states.}  We call a state maximal information or
        \textit{pure} if the corresponding isotropic vector space $V$ has the
        maximal dimension $n$, otherwise we call the state \textit{mixed}
        \cite{SpekkensFoundations2016}.  

\paragraph{Relation to stabilizer formalism.}    
        The above definition of states is reminiscent of the
        stabilizer formulation. In the case of $d = 2$, we
        can consider the functionals $\vec{f} \in V$ as the
        stabilizers, and the valuations as a generalization of the sign
        of the stabilizer. For example, if we know the value of the
        functional $f(\vec{m}) = q$ to be $1$, then we know that all
        possible ontic states are of the form $\begin{pmatrix} q
        \\[6pt]
            p \end{pmatrix} = \begin{pmatrix} 1 \\[6pt]
        b \end{pmatrix}$ with $b \in \{0,1\}$. We can assign the four
        possible ontic states $\begin{pmatrix} 0 \\[6pt]
        0 \end{pmatrix}, \begin{pmatrix} 1 \\[6pt]
        0 \end{pmatrix}, \begin{pmatrix} 0\\[6pt]
        1 \end{pmatrix}, \begin{pmatrix} 1 \\[6pt]
        1 \end{pmatrix}$ in this formalism to the ontic states
        ${1,2,3,4}$ in the previous formalism. Then the 
        epistemic state is $2 \lor 4$ which is stabilized by
        $-\mathcal{X}$. If the valuation of $f(\vec{m}) = q$ were $0$, then
        the state would be stabilized by $\mathcal{X}$. 
        
               \begin{definition}[Correspondence between stabilizer and  general formalisms]
        Let $d = 2$.
         We say  a stabilizer $g = \alpha p_1 \otimes ... \otimes p_n$ corresponds to an observable $\vec f$ (or vice-versa)
         if the entries of the observable  $\{f_{2j-1}, f_{2j}\}_{j \in \{1,...,n\}}$ are related to the stabilizer's individual qubits' observables $\{p_j\}_{j \in \{1,...,n\}}$ as 
        \begin{align}
            (f_{2j-1}, f_{2j} )= (0,0) \iff p_j = \mathbb{1}_2 ,\\
            (f_{2j-1}, f_{2j} )= (0,1)  \iff p_j = \mathcal{Z}, \\
            (f_{2j-1}, f_{2j} )= (1,0)  \iff p_j = \mathcal{X}, \\
            (f_{2j-1}, f_{2j} )= (1,1)  \iff p_j = \mathcal{Y},
        \end{align}
        and $\alpha \in \{-1,1\}$. 
        
        In addition, we say that a stabilizer state $\mathcal{S} = \langle g_1,...,g_N \rangle$ corresponds 
        to a generalized state $(V, \vec v)$  (or vice-versa) if and only if for each observable $\vec f \in V$ there is  a corresponding stabilizer $g \in \mathcal{S} = \alpha p_1 \otimes ... \otimes p_n$ satisfying $\alpha = (-1)^{\vec f^T \vec v}$ (that is, $\vec f^T \vec v = 0 \iff \alpha = 1$ and $\vec f^T \vec v = 1 \iff  \alpha = -1$).
\end{definition}
        
        \begin{restatable}[Correspondence between generalized and stabilizer formalisms is sound]{lemma}{lemmaStateCorrespondenceGentab}\label{toy:gen:stabcorrespondence}
             If $d=2$, each general state $(V,\vec v)$ has a corresponding stabilizer state and each stabilizer state $\mathcal{S}$ has a corresponding general state. Furthermore, if variables $\vec f$ and $\vec g$ correspond to stabilizers $g',g''$ respectively, then $\vec f + \vec g$ corresponds to $g' g''$. 
        \end{restatable}

   \paragraph{Remark.} It has been stated that for a valid state $(V, \vec v)$, the valuation
        $\vec{v}$ should be chosen in $V$ \cite{SpekkensFoundations2016,Catani_2017}; however, we found this to be unjustified. If we consider the
        case $d = 2$ and the isotropic space $V =
        \langle \vec f \rangle$, with $\vec f =p+q= \begin{pmatrix} 1 \\[6pt]
        1\end{pmatrix}$,  then the only two vectors in $V$ are $\vec v_0 = \begin{pmatrix} 0 \\[6pt]
        0\end{pmatrix}$  and $\vec v_1 = \begin{pmatrix} 1 \\[6pt]
        1\end{pmatrix}$, both of which have valuation $0$, as $\vec f^T \vec v_j \mod 2 = 0,$ $j=0,1$. There is no reason why there shouldn't exist a state satisfying  $p + q =1 $, but the
        only valuation vectors which produce this outcome (like $\begin{pmatrix} 1 \\[6pt]
        0\end{pmatrix}$) lie outside
        of $V$.  Thus, if we restrict the valuation vector $\vec{v}$ to
         lie in $V$, we unnecessarily limit the possible  number of
        outcomes of observables.

\newpage
\section{Measurements}
\label{sec:measurements}

\subsection{Original toy theory}

A measurement partitions the ontic state space into valid epistemic
            states, called the measurement basis.\footnote{We will see later that the converse can be problematic: not all partitions result in intuitive measurements.} The outcome of the
            measurement is determined by the position of the ontic
            state and is the epistemic state compatible with the ontic
            state \cite{Spekkens_2007}. Let us illustrate the measurement process with an example. Consider the following epistemic state:
            \begin{equation} \label{toy:org:mmt_state}
                \begin{tabular}{|{c}@{}|@{}{c}|{c}|{c}|}
                    \hline
                \colorcell{white}&\colorcell{white}&\colorcell{white}&\colorcell{cyan}\\ \hline
                \colorcell{white}&\colorcell{white}&\colorcell{cyan}&\colorcell{white}\\ \hline
                \colorcell{white}&\colorcell{cyan}&\colorcell{white}&\colorcell{white}\\ \hline
                \colorcell{cyan}&\colorcell{white}&\colorcell{white}&\colorcell{white}\\ \hline
                \end{tabular}
            \end{equation}
            and the measurement (where different numbers correspond to the epistemic states of the measurement basis):
            \renewcommand{\arraystretch}{1}
            \begin{equation}\label{toy:org:mmt}
                \begin{tabular}{|M| M |M| M|}

                    \cline{1-4}
                    3&3 & 4& 4 \\
                    \cline{1-4}
                    3&3 & 4& 4 \\
                    \cline{1-4}
                    
                    1&1 & 2& 2 \\
                    \cline{1-4}
                    
                    1&1 & 2& 2 \\
                    \cline{1-4}

                \end{tabular}
            \end{equation}
            \renewcommand{\arraystretch}{0}

                  
           The only results compatible with the epistemic state are
           4 and 1. Suppose the ontic state
           is $(4 \times 4)$; then the measurement outcome has to be
           4.

        \paragraph{Measurement update rule.}
        Let us consider two cases: a measurement of a single system ($N=1$), and a measurement on  multiple systems ($N>1$).
        
            \subparagraph{N = 1.}
                Let us assume the system is in the epistemic state 
                \begin{equation}
                    (1\lor 2) = 
                    \begin{tabular}{|{c}@{}|@{}{c}|{c}|{c}|}
                    \hline
                \colorcell{cyan}&\colorcell{cyan}&\colorcell{white}&\colorcell{white}\\ \hline
                \end{tabular}
                \end{equation}
                we measure in the basis ($a$ and $b$ denote two different basis states)
                \renewcommand{\arraystretch}{1}
                \begin{equation}
                    \{(1 \lor 3),(2 \lor4)\} = 
                        \begin{tabular}{|M| M |M| M|}

                            \cline{1-4}
                            a&b & a& b \\
                            
                            \cline{1-4}

                        \end{tabular}
                \end{equation}
                \renewcommand{\arraystretch}{0}
                and the outcome is $a$. From this measurement result, we
                can conclude that the ontic state before the
                measurement was $1$. This seems like a violation of
                the knowledge balance principle. However, the knowledge
                balance principle does not say anything about the
                knowledge an observer can have about the past state of a system. In fact, only the updated  epistemic state needs
                to fulfill the knowledge balance principle. As in
                quantum mechanics, we require that repeated
                measurement yields the same result. These conditions
                can be fulfilled through the following principle \cite{Spekkens_2007}. A measurement
                gives an unknown disturbance to the ontic state:
                either the permutation $1 \leftrightarrow 2$ is
                applied or no permutation is applied. Under this
                update rule, the post-measurement epistemic state is
                \begin{equation}
                    (1 \lor 3) = 
                    \begin{tabular}{|{c}@{}|@{}{c}|{c}|{c}|}
                    \hline
                \colorcell{cyan}&\colorcell{white}&\colorcell{cyan}&\colorcell{white}\\ \hline 
                \end{tabular} \ \text{or} \ (2 \lor 4) = 
                    \begin{tabular}{|{c}@{}|@{}{c}|{c}|{c}|}
                    \hline
                \colorcell{white}&\colorcell{cyan}&\colorcell{white}&\colorcell{cyan}\\ \hline 
                \end{tabular}
                \end{equation}

            \subparagraph{N $>$ 1.}
                In the case of multiple systems the update rule is
               more complicated. If the measurement is a conjunction
               of measurements on a single subsystem the update rule
               for one subsystem can be used. As an example consider
               the measurement in \cref{toy:org:mmt} applied to the
               state \cref{toy:org:mmt_state} and assume the system is
               in the ontic state $4 \times 4$. The measurement can be
               decomposed into the measurement on the first system
               
               \renewcommand{\arraystretch}{1}
            \begin{equation}
                \begin{tabular}{|M| M |M| M|}

                    \cline{1-4}
                    1&1 & 2& 2 \\
                    \cline{1-4}
                    1&1 & 2& 2 \\
                    \cline{1-4}
                    
                    1&1 & 2& 2 \\
                    \cline{1-4}
                    
                    1&1 & 2& 2 \\
                    \cline{1-4}

                \end{tabular}\ ,
            \end{equation}
            \renewcommand{\arraystretch}{0}
                and a measurement on the second system
            \renewcommand{\arraystretch}{1}
            \begin{equation}
                \begin{tabular}{|M| M |M| M|}

                    \cline{1-4}
                    1&1 & 1& 1 \\
                    \cline{1-4}
                    1&1 & 1& 1 \\
                    \cline{1-4}
                    
                    2&2 & 2& 2 \\
                    \cline{1-4}
                    
                    2&2 & 2& 2 \\
                    \cline{1-4}

                \end{tabular} \ .
            \end{equation}
            \renewcommand{\arraystretch}{0}
               The ontic state requires the measurement outcome of the
               first measurement to be 2. This excludes
               that $1 \cdot 1$ and $2 \cdot 2$ are the ontic state
               and the permutation $3 \leftrightarrow 4$ on
               the first system is randomly applied or not. Therefore, the
               updated state after the first measurement is 

               \begin{equation}\label{toy:org:mmt3}
                \begin{tabular}{|{c}@{}|@{}{c}|{c}|{c}|}
                    \hline
                \colorcell{white}&\colorcell{white}&\colorcell{cyan}&\colorcell{cyan}\\ \hline
                \colorcell{white}&\colorcell{white}&\colorcell{cyan}&\colorcell{cyan}\\ \hline
                \colorcell{white}&\colorcell{white}&\colorcell{white}&\colorcell{white}\\ \hline
                \colorcell{white}&\colorcell{white}&\colorcell{white}&\colorcell{white}\\ \hline
                \end{tabular} \ .
            \end{equation}  

                The second measurement will have the outcome 1 and the
                permutation $3 \leftrightarrow 4$ is randomly applied
                or not on the horizontal system. However, this
                permutation does not change the epistemic state and
                the epistemic state does not change after this
                measurement. 
                
\paragraph{Restriction on valid states.}  This update rule leads to a restriction
                on which epistemic states are valid. Consider the state 

                \begin{equation}\label{toy:org:mmt4}
                    \begin{tabular}{|{c}@{}|@{}{c}|{c}|{c}|}
                    \hline
                \colorcell{white}&\colorcell{cyan}&\colorcell{white}&\colorcell{white}\\ \hline
                \colorcell{cyan}&\colorcell{white}&\colorcell{white}&\colorcell{white}\\ \hline
                \colorcell{cyan}&\colorcell{white}&\colorcell{white}&\colorcell{white}\\ \hline
                \colorcell{cyan}&\colorcell{white}&\colorcell{white}&\colorcell{white}\\ \hline
                \end{tabular} 
                \end{equation} 
               At first glance one would think that this state was a
               valid epistemic state, as all marginals are valid and
               the knowledge balance principle is fulfilled on the
               bulk. However, if one measures the vertical system in
               $\{(1 \lor 3),(2 \lor4)\}$ and the outcome is $(2
               \lor4)$ this update rule says the state should be
               updated to 
               \begin{equation}\label{toy:org:mmt5}
                \begin{tabular}{|{c}@{}|@{}{c}|{c}|{c}|}
                    \hline
                \colorcell{white}&\colorcell{cyan}&\colorcell{white}&\colorcell{cyan}\\ \hline
                \colorcell{white}&\colorcell{white}&\colorcell{white}&\colorcell{white}\\ \hline
                \colorcell{white}&\colorcell{white}&\colorcell{white}&\colorcell{white}\\ \hline
                \colorcell{white}&\colorcell{white}&\colorcell{white}&\colorcell{white}\\ \hline
                \end{tabular}
            \end{equation} 
            which violates the knowledge balance principle. One could
            also argue that then it should not be allowed to measure
            the vertical system in $\{(1 \lor 3),(2 \lor4)\}$.
            However, for each measurement on a single system one can
            find a state with valid marginals that would violate the
            knowledge balance principle in a similar manner. Thus, one
            agent could not measure the system in their possession.
            For this reason, we do not forbid the measurement but instead
            say that the state is invalid \cite{Spekkens_2007}.

\paragraph{Coarse-grained measurements.}
               We can distinguish between maximal 
               and non-maximal information measurements. A maximal information measurement is a
               measurement where the space is partitioned into maximal
               information (or pure) states. In this case the state is updated
               to the state in the measurement basis which the outcome
               corresponds to. A non-maximal information measurement, or coarse-grained measurement,
               has at least some states in the measurement basis are
               non-maximal information (or mixed) states. This means there may be
               multiple valid epistemic states which yield the same
               result under repeated measurement. In this case, the
               measurement update rule cannot be uniquely defined by
               the conditions we imposed above. In quantum mechanics, a
               similar problem also appears and the most common update
               rule is to project the original state into the subspace
               associated with the result. In the toy theory, an
               analogous update rule can be defined using the fidelity
               \cite{Spekkens_2007}.

                \begin{definition}[Fidelity {\cite{Spekkens_2007}}]
                    The fidelity of two epistemic states $E_a$,$E_b$
                    is the classical fidelity of the uniform
                    probability distributions $P_a$,$P_b$ over the
                    ontic basis of the states $E_a$,$E_b$
                    \begin{equation}
                        F(E_a,E_b) = \sum_{o \in O} \sqrt{P_a(o)}\sqrt{P_b(o)}.
                    \end{equation} 
                \end{definition}
                Spekkens defines the updated state as the epistemic state with the
                highest fidelity with the pre-measurement state that
                will lead to the same output under repeated
                measurement \cite{Spekkens_2007}.

\paragraph{Example.} For example,
                consider the measurement

                 \renewcommand{\arraystretch}{1}
            \begin{equation}
                \begin{tabular}{|M| M |M| M|}

                    \cline{1-4}
                    1&1 & 3& 2 \\
                    \cline{1-4}
                    1&1 & 2& 3 \\
                    \cline{1-4}
                    
                    3&2 & 1& 1 \\
                    \cline{1-4}
                    
                    2&3 & 1& 1 \\
                    \cline{1-4}

                \end{tabular} \renewcommand{\arraystretch}{0}
            \qquad
            \text{on  state } 
            \qquad
         \begin{tabular}{|{c}@{}|@{}{c}|{c}|{c}|}
                    \hline
                \colorcell{white}&\colorcell{white}&\colorcell{white}&\colorcell{white}\\ \hline
                \colorcell{cyan}&\colorcell{cyan}&\colorcell{white}&\colorcell{white}\\ \hline
                \colorcell{cyan}&\colorcell{cyan}&\colorcell{white}&\colorcell{white}\\ \hline
                \colorcell{white}&\colorcell{white}&\colorcell{white}&\colorcell{white}\\ \hline
                \end{tabular}\quad .
            \end{equation} 
            Suppose that the measurement outcome is 1. The
            possible post measurement states are those with ontic support in that partition, which are the valid epistemic states
            \renewcommand{\arraystretch}{0}
            \begin{equation}
\begin{tabular}{|{c}@{}|@{}{c}|{c}|{c}|}
                    \hline
                \colorcell{cyan}&\colorcell{cyan}&\colorcell{white}&\colorcell{white}\\ \hline
                \colorcell{cyan}&\colorcell{cyan}&\colorcell{white}&\colorcell{white}\\ \hline
                \colorcell{white}&\colorcell{white}&\colorcell{cyan}&\colorcell{cyan}\\ \hline
                \colorcell{white}&\colorcell{white}&\colorcell{cyan}&\colorcell{cyan}\\ \hline
                \end{tabular} \ ,\quad
                \begin{tabular}{|{c}@{}|@{}{c}|{c}|{c}|}
                    \hline
                \colorcell{white}&\colorcell{white}&\colorcell{white}&\colorcell{white}\\ \hline
                \colorcell{white}&\colorcell{white}&\colorcell{white}&\colorcell{white}\\ \hline
                \colorcell{white}&\colorcell{white}&\colorcell{cyan}&\colorcell{cyan}\\ \hline
                \colorcell{white}&\colorcell{white}&\colorcell{cyan}&\colorcell{cyan}\\ \hline
                \end{tabular} \ ,\quad
                               \begin{tabular}{|{c}@{}|@{}{c}|{c}|{c}|}
                    \hline
                \colorcell{cyan}&\colorcell{cyan}&\colorcell{white}&\colorcell{white}\\ \hline
                \colorcell{cyan}&\colorcell{cyan}&\colorcell{white}&\colorcell{white}\\ \hline
                \colorcell{white}&\colorcell{white}&\colorcell{white}&\colorcell{white}\\ \hline
                \colorcell{white}&\colorcell{white}&\colorcell{white}&\colorcell{white}\\ \hline
                \end{tabular} \ ,\quad
                 \begin{tabular}{|{c}@{}|@{}{c}|{c}|{c}|}
                    \hline
                \colorcell{white}&\colorcell{cyan}&\colorcell{white}&\colorcell{white}\\ \hline
                \colorcell{cyan}&\colorcell{white}&\colorcell{white}&\colorcell{white}\\ \hline
                \colorcell{white}&\colorcell{white}&\colorcell{white}&\colorcell{cyan}\\ \hline
                \colorcell{white}&\colorcell{white}&\colorcell{cyan}&\colorcell{white}\\ \hline
                \end{tabular} \ ,\quad
                \begin{tabular}{|{c}@{}|@{}{c}|{c}|{c}|}
                    \hline
                \colorcell{cyan}&\colorcell{white}&\colorcell{white}&\colorcell{white}\\ \hline
                \colorcell{white}&\colorcell{cyan}&\colorcell{white}&\colorcell{white}\\ \hline
                \colorcell{white}&\colorcell{white}&\colorcell{cyan}&\colorcell{white}\\ \hline
                \colorcell{white}&\colorcell{white}&\colorcell{white}&\colorcell{cyan}\\ \hline
                \end{tabular}\ ,
            \end{equation} 
            where the first is a mixed state, and all others are pure. Indeed the first state can be seen as either a mixture of the second and third or of the fourth and fifth states (see section~\ref{sec:mixtures} for a definition of mixtures). 
            According to Spekkens' update rule we have to choose the state
            with maximal fidelity to the pre-measurement state, which is the third one.\footnote{As we will see  when we generalize the measurement update for arbitrary dimensions (Theorems~\ref{toy:gen:measurementupdate} and~\ref{toy:gen:measurementupdateequivalence}), there is nothing special about fidelity as a particular measure of overlap between two states, at least where this task in concerned. Other measures (with different powers of the distributions instead of a square root) may lead to the same choice of final state.} 
            
\subsection{Stabilizer formalism}

\paragraph{Quantum stabilizer measurement.}
Suppose we want to measure the Pauli
          operator $g$ on a quantum state by the stabilizer group $S$. If $\pm
          g \in S$ the outcome is deterministically $\pm 1$,
          otherwise with equal probability the outcome is $\pm 1$. In
          the latter case, the stabilizer group needs to be updated.
          As the measurement determined the value of $g$, $g$ itself is added
          to the stabilizer group. However, this could lead to an
          invalid stabilizer. Therefore, if there is at least one
          stabilizer $h$ in the generating set of $S$  that
          anticommutes with $g$, we multiply all other anticommuting
          generators with $h$ and replace $h$ with $g$. For example,
          let $S = span\{X_1 , X_2\}$ be a stabilizer group
          corresponding to the state $\ket{+}\ket{+}$. Let us measure the
          observable $Z_1$. Given that  $Z_1 \notin S$ the measurement
          outcome is random; suppose that it is $-1$. Since
          $X_1$ anticommutes  with $Z_1$, we remove $X_1$ and replace
          it with $-Z_1$. Thus, after the measurement the state has
          the stabilizer group $S = span\{-Z_1,X_2\}$, corresponding
           to $\ket{1}\ket{+}$ \cite{Gottesman_1997}. 

\paragraph{Toy theory: binary measurements. } In the epirestricted description of the toy theory we
 have argued that measurements are a partition of all
                ontic states into valid epistemic states. A toy
                observable $g$ then partitions the ontic states into
                the epistemic states $span\{g\}$ and $span\{-g\}$.
                The same procedure as in the quantum stabilizer formalism
                can be applied to update the epistemic state \cite{Pusey_2012}.
                We list all generators of the original state, remove the 
                  first generator $h$ that does not commute with $g$, multiply all other generators that do not commute with $g$ by $h$, and finally add either $g$ or $-g$ (depending on the outcome) to the list of generators. 
                  
\paragraph{Measurements on two elementary systems.} However, not all
                partitions are of the form $span\{g\}$ and
                $span\{-g\}$, in particular when a measurement has more than two possible outcomes. Consider the partition: 
                \renewcommand{\arraystretch}{1}
                \begin{equation}
                    \begin{tabular}{|M| M |M| M|}

                        \cline{1-4}
                        1&1 & 3& 2 \\
                        \cline{1-4}
                        1&1 & 2& 3 \\
                        \cline{1-4}
                        
                        3&2 & 4& 4 \\
                        \cline{1-4}
                        
                        2&3 & 4& 4 \\
                        \cline{1-4}

                    \end{tabular} \ .
                \end{equation}
                This partition corresponds to the following stabilizer
                states:  
                \begin{align}
                    \begin{split}
                        span\{\mathcal{X}_1 \mathcal{X}_2,\mathcal{Z}_1 \mathcal{Z}_2\}, \quad 
                       span\{-\mathcal{X}_1\mathcal{X}_2,\mathcal{Z}_1\mathcal{Z}_2\}, \quad 
                        span\{\mathcal{Z}_1,-\mathcal{Z}_2\}, \quad
                        span\{-\mathcal{Z}_1,\mathcal{Z}_2\}.
                    \end{split}
                \end{align}
                We can implement this measurement by first measuring
                $\mathcal{Z}_1 \mathcal{Z}_2$, which corresponds to
                the partition:
                
                \begin{equation}
                    \begin{tabular}{|M| M |M| M|}

                        \cline{1-4}
                        1&1 & 2& 2 \\
                        \cline{1-4}
                        1&1 & 2& 2 \\
                        \cline{1-4}
                        
                        2&2 & 1& 1 \\
                        \cline{1-4}
                        
                        2&2 & 1& 1 \\
                        \cline{1-4}

                    \end{tabular}
                \end{equation}
                \renewcommand{\arraystretch}{0}
                If the outcome is 1, we measure
                $\mathcal{X}_1 \mathcal{X}_2$, otherwise we measure the
                stabilizer $\mathcal{Z}_1$. In this way, we can implement
                more complicated partitions.
                
\paragraph{Measurements on more systems.} However, the above procedure
                does not work in general for more than two elementary systems. For example,  the partition
                \begin{align}
                    \begin{split}
                     span\{\mathcal{Z}_1,\mathcal{Z}_2,\mathcal{X}_3\} , \quad
                        span\{-\mathcal{Z}_1,\mathcal{X}_2,\mathcal{Z}_3\} , \quad
                        span\{\mathcal{X}_1,-\mathcal{Z}_2,-\mathcal{Z}_3\} ,\quad
                        span\{\mathcal{Z}_1,\mathcal{Z}_2,-\mathcal{X}_3\} , \\
                        span\{-\mathcal{Z}_1,-\mathcal{X}_2,\mathcal{Z}_3\} , \quad
                        span\{-\mathcal{X}_1,-\mathcal{Z}_2,-\mathcal{Z}_3\}, \quad
                        span\{\mathcal{Z}_1,-\mathcal{Z}_2,\mathcal{Z}_3\} , \quad
                        span\{-\mathcal{Z}_1,\mathcal{Z}_2,-\mathcal{Z}_3\}
                    \end{split}
                \end{align}
                is a valid measurement but does not have a non-trivial toy observable that is
                common to all the eight states.\footnote{Interestingly, the
                ontic states are partitioned into product states, but
                they cannot be distinguished using only local
                measurements. In quantum mechanics this property is
                known as \emph{non-locality without entanglement}~\cite{Bennett_1999};
                however, the toy theory is a
                local hidden variable theory, indicating that
                non-locality is not necessary~\cite{Pusey_2012}.}
                In cases like this, there is no easy way to find the post-measurement state.  We will see in the next section that this is because maybe not all partitions should be considered valid measurements.

\subsection{Generalizing to arbitrary dimensions}

 \paragraph{Valid measurements.}
        The set of valid measurements must respect the principle
        of classical complementarity. Therefore, we can only jointly
        measure variables that form an isotropic subspace or
        submodule $V_{\pi}$ with outcomes $\vec{v}_{\pi} \in \Omega$, in analogy to commuting operators in quantum theory. 
        The probability to get outcome $\vec{v}_{\pi}$ if the system
        is in an ontic state $\vec{m}$ is given by \cite{SpekkensFoundations2016}
        \begin{equation}
            \xi(\vec{v}_{\pi}|\vec{m}) = \delta_{V_{\pi}^{\perp} + \vec{v}_{\pi}}(\vec{m}).
        \end{equation}
        Intuitively, this means that we can only obtain measurement outcomes that are compatible with the ontic state of the system. 
         We can denote a measurement $V_{\pi}$ and its outcome
        $\vec{v}_{\pi}$ by the pair $(V_{\pi}, \vec{v}_{\pi})$. With
        this conditional probability distribution we can calculate the
        probability for a measurement outcome given the epistemic
        state $(V,\vec{v})$  \cite{SpekkensFoundations2016},
        \begin{equation}
            P(\vec{v}_{\pi}|(V,\vec{v})) = \sum_{\vec{m} \in \Omega} \xi(\vec{v}_{\pi}|\vec{m}) \  \mu_{(V,\vec{v})}(\vec{m}).
        \end{equation}
        \paragraph{Measurement update.}
        The discussion above allows to find the probability for measurement
        outcome, but does not tell us how to update a state $(V,v)$
        after the measurement $V_{\pi}$ with the outcome $\vec{v}_{\pi}$. In
        the past, there was an attempt at defining a measurement
        update which differentiated between the prime and non-prime
        case \cite{Catani_2017}. Here we unify these two cases and
        simplify the treatment of the measurement update. Starting from the principle that repeated measurements should result in the same outcome, we obtain the following result (proof in \cref{appendix:toyproofs:measurements}).
        
        \begin{restatable}[Measurement update rule (generalised)]{theorem}{toygenmeasurementupdate}\label{toy:gen:measurementupdate}
        
            When the epistemic state $(V,\vec{v})$ is subjected to a
            measurement $V_{\pi}$, and an outcome $\vec{v}_{\pi}$ is
            obtained, the epistemic state is updated to $(V', \vec v')$, where
         \begin{equation}
                V' = V_{\pi} \oplus V_{\text{commute}},
            \end{equation}
            \begin{equation} 
                \vec{v}' \in (V_{\pi}^{\perp} +  \vec{v}_{\pi}) \cap (V^{\perp}_{\text{commute}} +
                \vec{v}),
            \end{equation}
        and 
        \begin{align}
            V_{\text{commute}}  
            = \{\vec f \in V: \ [\vec f, \vec f_\pi] = 0, \ \forall \ \vec f_\pi \in V_\pi  \} 
            \subseteq V
        \end{align}
         is the subset of $V$ that commutes with $V_\pi$. 
        \end{restatable}

  \paragraph{Simple example.} 
  Let us consider the state 
  \begin{equation}
      \left(V =  \left\langle\begin{pmatrix} 1 \\[6pt]
        0 \end{pmatrix} \right\rangle, \vec{v} = \begin{pmatrix} 1 \\[6pt]
        0 \end{pmatrix}\right)  \iff \begin{tabular}{|{c}@{}|@{}{c}|{c}|{c}|}
                    \hline
                \colorcell{white}&\colorcell{cyan}&\colorcell{white}&\colorcell{cyan}\\ \hline
                \end{tabular}
  \end{equation}
  and the measurement
  
  \begin{equation}
  V_{\pi} = \left\langle \begin{pmatrix} 0 \\[6pt]
        1 \end{pmatrix} \right \rangle \iff
         \renewcommand{\arraystretch}{1}
      \begin{tabular}{|M| M |M| M|}

                            \cline{1-4}
                            a&a & b& b \\
                            
                            \cline{1-4}

                        \end{tabular}.
  \end{equation}
   \renewcommand{\arraystretch}{0}
    Then, $V_{commute}$ is given by ${ \left\langle\begin{pmatrix} 0 \\[6pt]
        0 \end{pmatrix} \right\rangle}$, as no non-zero vector in $V$ commutes with a vector in $V_{\pi}$. Suppose we get the measurement outcome $\vec v_{\pi}= \begin{pmatrix} 1 \\[6pt]
        1 \end{pmatrix}$. Therefore, the updated state is
        \begin{equation}
      \left(V_{\pi}\oplus V_{commute} = V_{\pi}, \vec{v} = \begin{pmatrix} 1 \\[6pt]
        1 \end{pmatrix}\right) \iff \begin{tabular}{|{c}@{}|@{}{c}|{c}|{c}|}
                    \hline
                \colorcell{white}&\colorcell{white}&\colorcell{cyan}&\colorcell{cyan}\\ \hline
                \end{tabular}.
  \end{equation}

  \paragraph{Complex example.} 
  Even though the last measurement we considered in the stabilizer formalism cannot be represented by an isotropic subspace $V_{\pi}$ with outcomes $\vec v_{\pi} \in \Omega$, we can still describe the update of state after the measurement, as it is a partition of state space. Let the outcome of the measurement be the state corresponding to $span\{\mathcal{Z}_1,\mathcal{Z}_2,\mathcal{X}_3\}$:
  \begin{equation}
     \left(W = \left\langle \begin{pmatrix} 0 \\[6pt]
        1 \\[6pt] 0  \\[6pt] 0 \\[6pt] 0 \\[6pt] 0\\[6pt] \end{pmatrix},  
        \begin{pmatrix} 0 \\[6pt]
        0 \\[6pt] 0  \\[6pt] 1 \\[6pt] 0 \\[6pt] 0\\[6pt] \end{pmatrix},  
        \begin{pmatrix} 0 \\[6pt]
        0 \\[6pt] 0  \\[6pt] 0 \\[6pt] 1 \\[6pt] 0\\[6pt] \end{pmatrix} \right\rangle, \vec w =  \begin{pmatrix} 0 \\[6pt]
        0 \\[6pt] 0  \\[6pt] 0 \\[6pt] 0 \\[6pt] 0\\[6pt] \end{pmatrix} \right).
  \end{equation}
  Then, we can use the before-defined measurement update, as the measurement update is just dependent on the outcome state. 
  Let us assume that the state before the measurement is
  \begin{equation}
     \left(V = \left\langle \begin{pmatrix} 0 \\[6pt]
        1 \\[6pt] 0  \\[6pt] 0 \\[6pt] 0 \\[6pt] 0\\[6pt] \end{pmatrix},  
        \begin{pmatrix} 0 \\[6pt]
        0 \\[6pt] 0  \\[6pt] 1 \\[6pt] 0 \\[6pt] 0\\[6pt] \end{pmatrix},  
        \begin{pmatrix} 0 \\[6pt]
        0 \\[6pt] 0  \\[6pt] 0 \\[6pt] 0 \\[6pt] 1\\[6pt] \end{pmatrix} \right\rangle, \vec v =  \begin{pmatrix} 0 \\[6pt]
        0 \\[6pt] 0  \\[6pt] 0 \\[6pt] 0 \\[6pt] 0\\[6pt] \end{pmatrix} \right).
  \end{equation}
  The vector space $V_{commute}$ is given by 
  \begin{equation}
      V_{commute} = \left\langle \begin{pmatrix} 0 \\[6pt]
        1 \\[6pt] 0  \\[6pt] 0 \\[6pt] 0 \\[6pt] 0\\[6pt] \end{pmatrix},  
        \begin{pmatrix} 0 \\[6pt]
        0 \\[6pt] 0  \\[6pt] 1 \\[6pt] 0 \\[6pt] 0\\[6pt] \end{pmatrix} \right\rangle
  \end{equation}
  and the vector $\vec v'$ can be chosen to be $\vec v$.
  To sum up, the update state is 
  \begin{equation}
     \left(V = \left\langle \begin{pmatrix} 0 \\[6pt]
        1 \\[6pt] 0  \\[6pt] 0 \\[6pt] 0 \\[6pt] 0\\[6pt] \end{pmatrix},  
        \begin{pmatrix} 0 \\[6pt]
        0 \\[6pt] 0  \\[6pt] 1 \\[6pt] 0 \\[6pt] 0\\[6pt] \end{pmatrix},  
        \begin{pmatrix} 0 \\[6pt]
        0 \\[6pt] 0  \\[6pt] 0 \\[6pt] 1 \\[6pt] 0\\[6pt] \end{pmatrix} \right\rangle, \vec v =  \begin{pmatrix} 0 \\[6pt]
        0 \\[6pt] 0  \\[6pt] 0 \\[6pt] 0 \\[6pt] 0\\[6pt] \end{pmatrix} \right).
  \end{equation}

\paragraph{Relation to original formalism.}
The measurement in the general formalism suffers from the same problem as the measurement in the stabilizer formalism; that is, not all partitions of the ontic space can be formalized as a measurement.
For the partitions that can be seen as a measurement $V_{\pi}$, the update rule is the same as in the original formulation.  
The generalized approach of this section is to be understood as more fundamental, and the original epirestricted formalism as a special case for small dimensions \footnote{This view is also supported by Spekkens (private correspondence).}.   Therefore, not all partitions should be considered valid measurements. For example, it is not clear whether allowing for all partitions as measurements would result in a theory equivalent to a quantum sub-theory, like stabilizer quantum mechanics. At the very least, we would no longer be able to interpret it as measuring a set of commuting observables $\vec f_1, \dots, \vec f_n$.
  
\begin{restatable}[Generalised measurement update rule reduces to the original formalism ]{theorem}{toygenmeasurementupdateequivalence}\label{toy:gen:measurementupdateequivalence}
For the case $d=2$, the measurement update rule  of \ref{toy:gen:measurementupdate} reduces to the measurement update rule in the original formalism.
\end{restatable}

\paragraph{Coarse-grained observables.} If $d$ is not prime and the greatest common divisor of the components $f_i$ of $\vec{f}$ is not $1$,
the variable cannot attain all values in $\mathbb{Z}_d$. We
call such variables \emph{coarse-grained}. For example, consider the case
$d = 4$ and $V = \langle\begin{pmatrix} 2 \\[6pt]
2 \end{pmatrix}\rangle$. Then there exist no two numbers $a,b
\in Z_a$ such that $\begin{pmatrix} 2 \\[6pt]
        2 \end{pmatrix}^T \begin{pmatrix} a \\[6pt]
        b \end{pmatrix} = 1$ \cite{Catani_2017}. Note that if the
        greatest common divisor of the components $f_i$ of an observable
        $\vec{f}$ is $1$ then such a variable can attain all values
        by  B\'ezout's lemma \cite{Bezout_1779,Catani_2017}. We call these variables \emph{fine-grained}.
        In the previous formulations the coarse-grained variables
        caused a problem when defining the measurement update rules
        \cite{Catani_2017}. However, in our more general formulation given above, coarse-grained variables are not problematic and can be observed.

\newpage
\section{Superpositions and mixtures}
\label{sec:mixtures}
A key feature of quantum theory is the existence of coherent superpositions and incoherent mixtures of pure states. These can also be found in Spekkens' toy theory, but one needs to be careful: if we mix or superpose arbitrary states, we risk violating the knowledge balance principle. 
Examples and intuitions for these restrictions are given in section~\ref{toy:gen:sup_mix} (generalized theory).  

Before we dig in, we would like to remark on two aspects of superpositions in quantum theory. The first one is that any state can be seen as either a diagonal state or a superposition, depending on the choice of basis; this aspect carries over to the toy theory, where every state can be seen as a superposition of elements of a different basis. The second aspect is that any number of quantum states can be superposed with any coefficients (provided we normalize the resulting state); for example $3 \ket0 + 2i \ket+$ is a valid, if unnormalized, quantum state. This is the aspect of superposition that is not always true in the toy theory, precisely because arbitrary superpositions could violate the knowledge-balance principle. The same issue emerges for probabilistic mixtures of states. 

\subsection{Original toy theory}
\label{toy:org:supmix}

\paragraph{Superpositions.}
In the original formulation, there is no definition for
                superposition of epistemic states on more than one
                system.\footnote{As we will see later, in the stabilizer formalism coherent
                superpositions can be defined for certain pairs of
                states composed of an arbitrary number of elementary
                systems \cite{Pusey_2012}.} For a single system, four
                different coherent superpositions can be defined, under some constraints on the pure states involved
                \cite{Spekkens_2007}.
                \begin{definition}[Superpositions \cite{Spekkens_2007}]
                Let $(a \lor b)$ and $(c \lor d)$ be two epistemic
                states with disjoint ontic bases and written such that
                $a < b$ and $c <d$. Then, we can define four coherent
                superpositions
                \begin{align}
                    (a \lor b) +_1 (c \lor d) &:= (a \lor c) \\[5pt]
                    (a \lor b) +_2 (c \lor d) &:= (b \lor d) \\[5pt]
                    (a \lor b) +_3 (c \lor d) &:= (b \lor c) \\[5pt]
                    (a \lor b) +_4 (c \lor d) &:= (a \lor d).
                \end{align} 
                \end{definition}
                
                \paragraph{Quantum analogy.}
                Comparing the resulting state of superposition with
                the analogous qubit state it becomes clear that these
                four different superpositions correspond to qubit
                superpositions $|a \rangle + e^{i \phi} |b\rangle$
                with different relative phases $\phi$
                \cite{Spekkens_2007}
                \begin{align}
                    \begin{split}
                        +_1 &\iff \phi = 0 \\
                        +_2 &\iff \phi = \pi \\
                        +_3 &\iff \phi = \pi/2 \\
                        +_4 &\iff \phi = 3\pi/2 
                    \end{split} 
                \end{align}
                Curiously, the analogy between the qubit and the
                single system epistemic state breaks down under
                superpositions \cite{Spekkens_2007}. Consider the
                superposition 
                
                \begin{equation}
                    (1 \lor 3) +_3 (2 \lor 4) = (2 \lor 3) \iff |+i \rangle.
                \end{equation}
                 However, the analogous qubit superposition gives a
                 different result\cite{Spekkens_2007}:
                 \begin{equation}
                    |+\rangle + e^{i \pi/2} |-\rangle = e^{i\pi/4}|-i\rangle.
                \end{equation}

            \paragraph{Mixtures.}
               Incoherent mixtures can be 
                defined for any number of subsystems, but there is
                still a requirement on which states can be mixed
                \cite{Spekkens_2007}: 
                \begin{definition}[Mixtures \cite{Spekkens_2007}]\label{def:org:mix}
                    Let $E_a$ and $E_b$ be two epistemic states with
                    disjoint ontic bases $O_a$ and $O_b$, such that
                    the union of the ontic basis is a valid epistemic
                    state, then the mixture $E_a + E_b$ between these two states
                    is the state with the ontic basis $O_a \cup O_b$.
                \end{definition}
                From this definition two questions arise. Firstly, it
                is clear from this definition that any state that
                can be written as the mixture of pure states is a
                mixed state. However, the converse is not obviously
                true: are all toy mixed states mixtures of pure states?  We will show in the stabilizer formalism that this is indeed true. Similarly to quantum
                mechanics, this decomposition is not unique as $(1
                \lor 2 \lor 3 \lor 4)$ has multiple decompositions:
                \begin{align}
                    \begin{split}
                        (1 \lor 2 \lor 3 \lor 4) &= (1 \lor 2) + (3 \lor 4) \\
                        (1 \lor 2 \lor 3 \lor 4) &= (1 \lor 3) + (2 \lor 4) \\
                        (1 \lor 2 \lor 3 \lor 4) &= (1 \lor 4) + (2 \lor 3).
                    \end{split}
                \end{align}
                Secondly, the
                definition requires that the union of two ontic bases
                forms again an ontic basis, but it is in general hard
                to check if a state is valid due to the inductive
                nature of the epistemic states. The stabilizer
                formalism can give us a simpler condition \cite{Pusey_2012,Spekkens_2007}.

\subsection{Stabilizer formalism}

            \paragraph{Mixtures.}  Mixtures can only be defined for states where the
            corresponding stabilizer groups are \emph{rephasings} of each
            other. The stabilizer group $S'$ is a rephasing of $S$ if
            and only if for any $g \in S$, $\pm g \in S'$. For two
            stabilizer groups $S$, $S'$ which are rephasings of each
            other the mixture $\frac{1}{2}(\rho_S + \rho_{S'})$ is
            given by the state corresponding to the stabilizer group
            $T = S \cap S'$. 
             This definition allows us to state the following theorem (proof in the appendix).
            \begin{restatable}[Mixed states as convex combinations]{theorem}{toystabmix}\label{toy:stab:mix}
                Any mixed state with stabilizer
                group $\mathcal{S} = span\{g_1, ...,g_{N-k}\}$, $k
                \leq N$ can be written as the mixture of $2^k$
                pure states.
            \end{restatable}
            
            \subparagraph{Example.} Consider the stabilizer
            groups $S = span\{-Z_1\}$ and $S' = span\{Z_1\}$,
            corresponding to the states $|1 \rangle$ and $|0\rangle$
            respectively. The state generated by the rule above has
            the stabilizer group $T = \{\mathbb{1}_2 \}$, which
            corresponds to the state $\frac{1}{2}(|0\rangle \langle
            0|+|1\rangle \langle 1|)$ \cite{Pusey_2012}.

\paragraph{Superpositions.} 
            For superpositions, we allow orthogonal pure states
            which are different and a rephasing of each other. 
            Let $\mathcal{S}$ and $\mathcal{S}'$ be two such pure states and let $\mathcal{S}$ be spanned by $g_1, ..., g_N$ which are independent generators. 
            Then in general $\mathcal{S}'$ is spanned by the same generators, but with potentially different signs.  These two states are either identical or there is at least one generator $g_i \in \mathcal{S}$ where $-g_i \in \mathcal{S}'$. We can multiply all generators with $g_j \neq g_i$ of $\mathcal{S}$ such that $- g_j \in \mathcal{S}'$ with $g_i$ for $\mathcal{S}$ and $-g_i$ for $\mathcal{S}$. 
            The resulting generators are still independent and span the same state. 
            Therefore, in general we can consider two states which are a rephasing of each other to be of the following form
            $\mathcal{S} = span\{g_1, ..., g_N\}$ and
            $\mathcal{S}' = span\{g_1, ..., -g_N\}$ where only the sign of the last generator changes and
            $g_1, ..., g_N$ are independent generators. Then the
            coherent superposition of the two states represented by these groups is given by the stabilizer group $T
            = span\{g_1, ...,g_{N-1}, h\}$ where $h$ is such that $h \notin S
            \cup S'$. Such an $h$ exists by a similar argument as in
            the stabilizer formalism. In every case there are four
            different $h$ that can be chosen. In the case of a single
            system, they correspond directly to the four different
            superpositions defined in \cref{toy:org:supmix} \cite{Pusey_2012}. 
            
            \subparagraph{Example.} Consider $S =
            span\{Z_1Z_2,Z_2\}$ and $S'= span\{Z_1Z_2,-Z_2\}$; these
            correspond to the states $|00\rangle$ and $|11\rangle$
            respectively. By our definition the superposition would be
            $ span\{Z_1Z_2,h\}$ with $h = \pm X_1X_2$ or $h = \pm
            Y_1Y_2$. These superpositions  correspond to
            the states $\frac{1}{\sqrt{2}}(|00\rangle \pm |11\rangle)$
            and $\frac{1}{\sqrt{2}}(|00\rangle \pm i|11\rangle)$
            respectively \cite{Pusey_2012}.

\subsection{Generalizing to arbitrary dimensions}
\label{toy:gen:sup_mix}

\paragraph{Mixtures.}
In the generalized version of the toy theory, it is easy to see why we cannot simply mix any two states. 
Recall that the ontic basis of a valid epistemic state  $(V,\vec{v})$ is the affine subvector space (or submodule) $ V^{\perp} + \vec{v}$.
If we were to mix two arbitrary states $(V_1,\vec{v}_1)$ and $(V_2,\vec{v}_2)$ by taking the union of their ontic basis, the resulting ontic basis of the mixture would be $(V_1^{\perp} + \vec{v}_1) \cup (V_2^{\perp} + \vec{v}_2)$. However, this is not necessarily a subvector space or submodule, and therefore may not represent a valid epistemic state. 

We will see that the conditions for valid mixtures result in the following restriction: we can only mix 
epistemic states that correspond to making the same observations but obtaining different results; we also have to impose that each  of these observables be either \emph{constant} across the the family of states to be mixed (that is, it has the same valuation for all the states) or \emph{totally unknown} (\emph{all possible outcomes} are present in the family of states). 

When we mix these states we will obtain a new state $(V, \vec w)$, for which we should still know the valuation of the constant variables (as they have the same outcome for all states in the family). As for the other variables, all we know is that all outcomes are possible; but this is the same as not knowing their valuation at all (hence the name ``totally unknown''). 
%
%
 Therefore, the mixed epistemic state $(V, \vec w)$ would not change if we removed all the \emph{totally unknown} observables from $V$. This implies that an alternative way to obtain the mixture in the first place is by taking any one of the original states $(V, \vec{v}_j)$ and simply remove the totally unknown variables from $V$, leaving only the valuation of the variables that are constant across the family. This procedure is akin to taking a partial trace over the non-constant observables. 
 We can formalize this discussion in the following definitions. 
 \begin{definition}[Constant and totally unknown observables] \label{def:toy:gen:constantunknown}
    Let  $\{(V, \vec{v}_j)\}_j$ be a family of epistemic states. We say that an observable $\vec f_i \in V$ is 
    \begin{itemize}
        \item \emph{constant} or \emph{known} if it has the same valuation for all states in the family, $\vec f_i^T \vec v_j = k_i, \quad \forall \ \vec v_j$;
        \item \emph{totally unknown} if for all possible outcomes $ \ k_{i, \ell}$ of $\vec f_i$, there exists a member of the family satisfying $ \vec f_i^T \vec v_j = k_{i, \ell}$;
        \item \emph{partially known} otherwise.
    \end{itemize}

 \end{definition}
 
 \begin{definition}[Mixture of states (generalized)]
    \label{def:toy:gen:mixtures}
    Let  $\{(V, \vec{v}_j)\}_j$ be a family of epistemic states such that each observable $\vec f_i \in V$ is either constant or totally unknown. 
    Then the \emph{epistemic mixture} of the states in the family is the state $(W,\vec{v}_1)$, where $$W = \left\langle \{ \vec f_i \in V: \vec f_i \text{ constant}\} \right\rangle.$$ 
 \end{definition}
    
The following lemma ensures that this definition results in the intuitive mixture of states.
    
\begin{restatable}[Soundness of generalized mixtures]{lemma}{toygenmixtures}\label{toy:gen:mixtures}
   Definition \ref{def:toy:gen:mixtures} is sound: it results in a valid epistemic state $(W, \vec w)$. The ontic basis of the mixture is the union of the ontic basis of each element of the family, 
   $$W^{\perp} + \vec{w} = \bigcup_{j} (V^{\perp} + \vec{v}_j).$$
   The choice of $\vec w = \vec v_1$ is arbitrary, as by definition  all  the $\vec{v}_j$ are all equivalent valuation vectors for the constant variables.
\end{restatable} 
 

If we would not require that \emph{all} different possible valuations could be attained for the non-constant variables, then we would know that some valuations of these variables were impossible. 
This would be more knowledge than not knowing the variables, and less than knowing them completely --- and such intermediate, partial knowledge cannot be described in the toy theory. 
 
\subparagraph{Example of an illegal mixture.} Suppose that, for $d=2$, we wanted to mix the states
        \begin{equation}
            \left(\left\langle\begin{pmatrix} 0 \\[6pt]
                1 \\[6pt] 0 \\[6pt] 0 \\[6pt]\end{pmatrix},\begin{pmatrix} 0 \\[6pt]
                0 \\[6pt] 0 \\[6pt] 1 \\[6pt]\end{pmatrix}\right\rangle, \begin{pmatrix} 0 \\[6pt]
                0 \\[6pt] 0 \\[6pt] 0 \\[6pt]\end{pmatrix} \right), \
            \left(\left\langle\begin{pmatrix} 0 \\[6pt]
                1 \\[6pt] 0 \\[6pt] 0 \\[6pt]\end{pmatrix},\begin{pmatrix} 0 \\[6pt]
                0 \\[6pt] 0 \\[6pt] 1 \\[6pt]\end{pmatrix}\right\rangle, \begin{pmatrix} 0 \\[6pt]
                1 \\[6pt] 0 \\[6pt] 0 \\[6pt]\end{pmatrix} \right),
        \end{equation}
        and
        \begin{equation}
            \left(\left\langle\begin{pmatrix} 0 \\[6pt]
                1 \\[6pt] 0 \\[6pt] 0 \\[6pt]\end{pmatrix},\begin{pmatrix} 0 \\[6pt]
                0 \\[6pt] 0 \\[6pt] 1 \\[6pt]\end{pmatrix}\right\rangle, \begin{pmatrix} 0 \\[6pt]
                0 \\[6pt] 0 \\[6pt] 1 \\[6pt]\end{pmatrix} \right).
        \end{equation}
        We can represent these states in phase space in the following
        grid, the same way it's done in the original toy theory~\footnote{This representation is helpful to connect the original
        formulation and the generalization for $d = 2$ \cite{SpekkensFoundations2016}.}:
        \renewcommand{\arraystretch}{1.1}
        \begin{equation}
            \begin{tabular}{c C C C C C}

                \cline{3-6}
                \multirow{4}{*}{$(p_1,q_1)$ \hspace{0.05cm} }& \multicolumn{1}{c|}{11 \hspace{0.1cm}}&
                \multicolumn{1}{c|}{ }&
                \multicolumn{1}{c|}{ }&\multicolumn{1}{c|}{ }&
                \multicolumn{1}{c|}{ }\\
                \cline{3-6}
                &\multicolumn{1}{c|}{10 \hspace{0.1cm}}&
                \multicolumn{1}{c|}{ }&
                \multicolumn{1}{c|}{ }&\multicolumn{1}{c|}{ }&
                \multicolumn{1}{c|}{ }\\
                \cline{3-6}
                &\multicolumn{1}{c|}{ 01 \hspace{0.1cm}}&
                \multicolumn{1}{c|}{ }&
                \multicolumn{1}{c|}{ }&\multicolumn{1}{c|}{ }&
                \multicolumn{1}{c|}{ }\\
                \cline{3-6}
                &\multicolumn{1}{c|}{ 00 \hspace{0.1cm}}&
                \multicolumn{1}{c|}{ }&
                \multicolumn{1}{c|}{ }&\multicolumn{1}{c|}{ }&
                \multicolumn{1}{c|}{ }\\
                \cline{3-6}
                &\multicolumn{1}{c}{ }&\multicolumn{1}{c}{ \hspace{0.03cm} 00 \hspace{0.03cm}} & \multicolumn{1}{c}{ \hspace{0.03cm} 01 \hspace{0.03cm}} & \multicolumn{1}{c}{ \hspace{0.03cm} 10 \hspace{0.03cm}}& \multicolumn{1}{c}{ \hspace{0.03cm} 11 \hspace{0.03cm}}\\ 
                \multicolumn{1}{c}{}&\multicolumn{1}{c}{}&\multicolumn{4}{c}{$(p_2,q_2)$}

            \end{tabular}.
        \end{equation}
        \renewcommand{\arraystretch}{0}
        In this representation, the above mixing would correspond to mixing the
        states 
        \begin{equation}
           \begin{tabular}{|{c}@{}|@{}{c}|{c}|{c}|}
                    \hline
                \colorcell{white}&\colorcell{white}&\colorcell{white}&\colorcell{white}\\ \hline
                \colorcell{white}&\colorcell{white}&\colorcell{white}&\colorcell{white}\\ \hline
                \colorcell{cyan}&\colorcell{cyan}&\colorcell{white}&\colorcell{white}\\ \hline
                \colorcell{cyan}&\colorcell{cyan}&\colorcell{white}&\colorcell{white}\\ \hline
                \end{tabular}\ , \quad
            \begin{tabular}{|{c}@{}|@{}{c}|{c}|{c}|}
                    \hline
                \colorcell{cyan}&\colorcell{cyan}&\colorcell{white}&\colorcell{white}\\ \hline
                \colorcell{cyan}&\colorcell{cyan}&\colorcell{white}&\colorcell{white}\\ \hline
                \colorcell{white}&\colorcell{white}&\colorcell{white}&\colorcell{white}\\ \hline
                \colorcell{white}&\colorcell{white}&\colorcell{white}&\colorcell{white}\\ \hline
                \end{tabular} \quad \text{and} \quad
            \begin{tabular}{|{c}@{}|@{}{c}|{c}|{c}|}
                    \hline
                \colorcell{white}&\colorcell{white}&\colorcell{white}&\colorcell{white}\\ \hline
                \colorcell{white}&\colorcell{white}&\colorcell{white}&\colorcell{white}\\ \hline
                \colorcell{white}&\colorcell{white}&\colorcell{cyan}&\colorcell{cyan}\\ \hline
                \colorcell{white}&\colorcell{white}&\colorcell{cyan}&\colorcell{cyan}\\ \hline 
                \end{tabular}\ .
        \end{equation}
        However, the union of these states is 
        \begin{equation}
        \begin{tabular}{|{c}@{}|@{}{c}|{c}|{c}|}
                    \hline
                \colorcell{cyan}&\colorcell{cyan}&\colorcell{white}&\colorcell{white}\\ \hline
                \colorcell{cyan}&\colorcell{cyan}&\colorcell{white}&\colorcell{white}\\ \hline
                \colorcell{cyan}&\colorcell{cyan}&\colorcell{cyan}&\colorcell{cyan}\\ \hline
                \colorcell{cyan}&\colorcell{cyan}&\colorcell{cyan}&\colorcell{cyan}\\ \hline
                \end{tabular}\ ,
    \end{equation}
    which is not a valid epistemic state. Therefore, their mixture cannot
    be valid.

        \paragraph{Superpositions.}
      
        In quantum mechanics, position and momentum eigenstates are
        connected via a Fourier transform, i.e 
        \begin{equation}
           \ket{q} = \frac{1}{2 \pi} \ \int_{\mathbb{R}} dp \ e^{ipq} \ \ket{p}.
        \end{equation}
        Thus, certain superpositions of momentum states are  position
        eigenstates. We can define superposition in the toy model in
        an analogous way, starting from the toy analogs for
        $|p\rangle$ and $|q\rangle$.
        
        In the continuous case we could have a vector space $V_p = \left\langle \begin{pmatrix} 0 \\[6pt]
        1 \\[6pt] \end{pmatrix}  \right\rangle$, which is the toy analog of $\ket{p}$. 
        Then we can superpose the states ${(V_p,\vec{v}_p)}_{p \in \mathbb{R}}$, where $\vec v_p =  \begin{pmatrix} p \\[6pt]
        p \\[6pt] \end{pmatrix} $.
        The resulting state of the superposition is the state $(V_q = \left\langle \begin{pmatrix} 1 \\[6pt]
        0 \\[6pt] \end{pmatrix}  \right\rangle, \vec v_q =  \begin{pmatrix} q \\[6pt]
        q \\[6pt] \end{pmatrix} )$. The valuation vector of the superimposed state can be chosen arbitrarily and corresponds to the different choices of $q$ in $e^{ipq}$ in the quantum superposition.

\begin{definition}[Superpositions over one observable]
\label{def:eqsuperpose}
        Let  $\{(V, \vec{v}_j)\}_j$ be a family of epistemic states with $V= \langle \vec f_1, \dots, \vec f_k \rangle  $ such that one of the  observables $\vec f_k$ is totally unknown, and all others are constant. 
        Let $\vec f$ be another observable such that $V':=\langle \vec f_1, \dots, \vec f_{k-1} \rangle \oplus \vec f $ is another isotropic subspace (not equal to $V$). 
        
        We call the new state $(V', \vec v_j)$, where $\vec v_j$ is one of the valuations in the original family, the \emph{superposition over $\vec f_k$} with a choice of valuation $\vec v_j$ and an additional observable $\vec f$.  These different choices correspond to different \emph{phases} in the superposition.
\end{definition}

        \begin{restatable}[{Soundness of generalized superpositions}]{lemma}{toygensuperpose}\label{toy:gen:superpose}
            For the continuous and discrete prime cases, definition~\ref{def:eqsuperpose} is sound.  That is, there is always an observable $\vec f$ such that $V'=\langle \vec f_1, \dots, \vec f_{k-1} \rangle \oplus \vec f $ is an isotropic subspace not equal to $V$. Furthermore, $(V', \vec v_j)$ is a valid epistemic state for every  valuation $\vec v_j$ in the original family. 
        \end{restatable}

        
        The definition of superpositions in the generalized formalism can be reduced to its stabilizer formalism definition for the case $d=2$, which is formally stated in the following lemma.
        
        \begin{restatable}[Reduction of general superpositions to the stabilizer case]{lemma}{toygensuperposestabilizer}\label{toy:gen:superposestabilizer}
            Definition~\ref{def:eqsuperpose} reduces to the one in the stabilizer formalism for the case $d=2$. 
            That is, the corresponding stabilizer state to the superposition of two general states is the superposition their corresponding stabilizer states, and the corresponding general state of the superposition of two stabilizer states is the superposition of the corresponding general states. 
        \end{restatable}
        
       \subparagraph{Example.}
       Let us consider the case where $d = 3, n=2$,  the  isotropic vector space 
       \begin{equation}
            V = \left\langle\begin{pmatrix} 1 \\[6pt]
                0 \\[6pt] 0 \\[6pt] 0 \\[6pt]\end{pmatrix},\begin{pmatrix} 0 \\[6pt]
                0 \\[6pt] 1 \\[6pt] 0 \\[6pt]\end{pmatrix}\right\rangle.
        \end{equation}
        and the three states with valuation vectors 
        \begin{align}
            \begin{split}
              \vec  v_{0} &= \begin{pmatrix} 0 \\[6pt]
                0 \\[6pt] 0 \\[6pt] 0 \\[6pt]\end{pmatrix}, 
           \vec v_{1} = \begin{pmatrix} 1 \\[6pt]
                0 \\[6pt] 0 \\[6pt] 0 \\[6pt]\end{pmatrix},
          \vec  v_{2} = \begin{pmatrix} 2 \\[6pt]
                0 \\[6pt] 0 \\[6pt] 0 \\[6pt]\end{pmatrix}.
            \end{split}
        \end{align}
        The first observable in $V$ is totally unknown for all states, while the second one is constant. Therefore, we can take the superposition over the three states. 
        The vector 
        \begin{equation}
             \vec f= \begin{pmatrix} 0 \\[6pt]
                1 \\[6pt] 0 \\[6pt] 0 \\[6pt]\end{pmatrix}
        \end{equation}
        commutes with the second observable generating $V$. Therefore,
        \begin{equation}
            V' = \left\langle\begin{pmatrix} 0 \\[6pt]
                1 \\[6pt] 0 \\[6pt] 0 \\[6pt]\end{pmatrix},\begin{pmatrix} 0 \\[6pt]
                0 \\[6pt] 1 \\[6pt] 0 \\[6pt]\end{pmatrix}\right\rangle.
        \end{equation}
        can be seen as the isotropic space of the superposition of the rows. All three states $(V', \vec v_0)$, $(V', \vec v_1)$ and $(V', \vec v_2)$ are valid superpositions.
        Note that there are other choices available for the new variable, for example 
        \begin{equation}
            \vec f'= \begin{pmatrix} 0 \\[6pt]
                0 \\[6pt] 0 \\[6pt] 1 \\[6pt]\end{pmatrix}
        \end{equation}
        also commutes with the second observable generating $V$ and $V'$. Hence,
        \begin{equation}
            V'' = \left\langle\begin{pmatrix} 0 \\[6pt]
                1 \\[6pt] 0 \\[6pt] 0 \\[6pt]\end{pmatrix},\begin{pmatrix} 0 \\[6pt]
                0 \\[6pt] 0 \\[6pt] 1 \\[6pt]\end{pmatrix}\right\rangle
        \end{equation}
        is also an isotropic space, and as such the superpositions $(V'', \vec v_0)$, $(V'', \vec v_1)$ and $(V'', \vec v_2)$ are all valid states.

\newpage
\section{Entanglement}
\label{sec:entanglement}

In quantum mechanics, the phenomenon of entanglement implies the existence of global states that cannot be written as a product of the states of subsystems, and demonstrates non-classical statistical relations between the subsystems~\cite{Einstein1935, Bell_1964}. As the toy theory mimics quantum theory to an extent, we are also able to establish the notion of entanglement there.  The proofs for new results in this section can be found in Appendix~\ref{appendix:toyproofs}.

\paragraph{Meaning of entanglement.}   One may wonder what it means to define ``entanglement'' in a local hidden variable theory. 
In Spekkens' toy theory, to have knowledge about an individual system is to have answers to some of the questions about its ontic state. To have knowledge of two systems being entangled is to be able to answer some questions on how these two systems are correlated. In other words, the information contained in a composite system can be divided into the information carried by individual systems, and the information contained in the correlations between observations made on individual subsystems~\cite{Brukner2001}. 
 Thanks to the knowledge balance principle, it is hence possible to either possess maximal knowledge about the states of subsystems -- which results in their epistemic states being uncorrelated; or about the correlations between them, which is described as a correlated epistemic state -- and an entangled one, if it is also a state of maximal knowledge. 
 To sum up, to be in an entangled state in the toy theory simply means to have maximal possible information about the correlations between systems. In particular, this is not at odds with the toy theory being an example of a local hidden variable theory.

\subsection{Original toy theory}
\label{subsec:entanglement-toy}
                To define entanglement in the epirestricted picture, we first need to define product
                states, i.e.\ non-entangled states.
                
                \begin{definition}[Product states]
                    We say that a state is of a product form between two systems $A$ and $B$ (or is a \emph{product state}) if it can be written as 
                    $$E_A \times E_B = \bigvee_{i=1}^{k_A} \bigvee_{j=1}^{k_B} (a_i
                \times b_j)  = (a_1
                \times b_1) \lor (a_1 \times b_2) \lor \dots \lor (a_2
                \times b_1) \lor \dots,$$ where
                $E_A = (a_1 \lor a_2 \lor \dots a_{k_A})$ and $E_B = (b_1
                \lor b_2 \lor \dots b_{k_B})$ are two valid epistemic states on systems $A$ and $B$ respectively.  
                \end{definition}
                
                \begin{restatable} [Product states are sound]{theorem}{toyproduct}
                 \label{toy:org:ent}
                Let $E_A = (a_1 \lor a_2 \lor \dots a_{k_A})$ and $E_B = (b_1
                \lor b_2 \lor \dots b_{k_B})$ be two valid epistemic states on two systems $A$ and $B$. Then the state $E_A \times E_B = \bigvee_{i=1}^{k_A} \bigvee_{j=1}^{k_B} (a_i
                \times b_j) $ is a valid epistemic state. 
                \end{restatable}

                Using this definition, we can define entanglement and correlations. 
                \begin{definition}[Entanglement]
                    We call a pure epistemic state defined on systems
                    $A$ and $B$
                    \emph{entangled} if it cannot be written as product of
                    epistemic states on $A$ and $B$. 
                    
                    We call a mixed state \emph{correlated} if it is not  of product form. We say a mixed state is \emph{entangled} if it cannot be written as a mixture of
                    product states.
                \end{definition}

            \paragraph{Examples.}    An example for a
                correlated state would be the state
                      
                
                \begin{equation}
                \begin{tabular}{|{c}@{}|@{}{c}|{c}|{c}|}
                    \hline
                \colorcell{white}&\colorcell{white}&\colorcell{cyan}&\colorcell{cyan}\\ \hline
                \colorcell{white}&\colorcell{white}&\colorcell{cyan}&\colorcell{cyan}\\ \hline
                \colorcell{cyan}&\colorcell{cyan}&\colorcell{white}&\colorcell{white}\\ \hline
                \colorcell{cyan}&\colorcell{cyan}&\colorcell{white}&\colorcell{white}\\ \hline
               
                \end{tabular}
              \end{equation}
This is the toy analog of the classically correlated state $\proj {00}{00} + \proj{11}{11}$.
                An example of an entangled state would be the state
                      
                \begin{equation}
                \begin{tabular}{|{c}@{}|@{}{c}|{c}|{c}|}
                    \hline
                \colorcell{white}&\colorcell{white}&\colorcell{white}&\colorcell{cyan}\\ \hline
                \colorcell{white}&\colorcell{white}&\colorcell{cyan}&\colorcell{white}\\ \hline
                \colorcell{white}&\colorcell{cyan}&\colorcell{white}&\colorcell{white}\\ \hline
                \colorcell{cyan}&\colorcell{white}&\colorcell{white}&\colorcell{white}\\ \hline
               
                \end{tabular}
              \end{equation}
                which can be seen as the toy theory analog of the
                Bell state $|00\rangle + |11\rangle$. This state can
                be generalized to states composed of more systems. In
                general these states are analogous to cat states $|0\rangle^{\otimes N}
                + |1\rangle^{\otimes N}$ and are given by 
             \begin{align}\label{toy:org:ent_states}
                    \begin{split}
                        &\lor_{x_1,...,x_{N-1} \in \{0,1\}}
                \left(\pi_1(x_1) \times ... \times \pi_{N-1}(x_{N-1}) \times \pi_N((x_1 + ...+ x_{N-1}
                )_{\text{mod 2}})\right) \lor \\ 
                &\lor_{x_1,...,x_{N-1} \in \{0,1\}}
                (\pi_1'(x_1)\times...\times\pi_{N-1}'(x_{N-1}) \times \pi_N'((x_1 + ...+ x_{N-1}
                )_{\text{mod 2}}))
                    \end{split}
                \end{align}
                were $\pi'_i$ and $\pi_i$ are injective functions
                which map from $\{0,1\} \to \{1,2,3,4\}$ and fulfill
                $\pi_i(\{0,1\} ) \cap \pi'_i(\{0,1\} ) = \emptyset$.
                We show in the stabilizer formalism that the
                states of this type are valid. For all different $\pi'_i$ and
                $\pi_i$, these states are equivalent under local
                transformations. Note that for systems with $N \geq 4$,
                \cref{toy:org:ent_states} is not the only allowed type of an entangled state.

\subsection{Stabilizer formalism}
\label{subsec:entanglement-stab}

\paragraph{Quantum pure state entanglement.} In quantum stabilizer formalism, pure state entanglement has a straightforward interpretation. Let $S$ be a stabilizer group of $n$ qubits, and $A$,$B$ -- two subsystems consisting of $k_A$ and $k_B = n-k_A$ qubits respectively. Then the stabilizer group can be split into: a local subgroup $S_A \cdot S_B =
            span\{S_A,S_B\}$, such that $S_A$ and $S_B$ only contain
            stabilizers of the systems $A$ and $B$ respectively; and
            the subgroup of the rest $S_{AB}$. It can be shown that
            the size of the minimal generating set of $S_{AB}$ is
            double the entanglement entropy of the state stabilized by
            $S$ \cite{Fattal2004}.
            
\paragraph{Toy pure state entanglement.}
            We can use the same idea as in quantum mechanics to
            characterize entanglement in the toy theory.
            \begin{restatable} [Characterization of pure entangled states]{theorem}{toystabprod}\label{toy:stab:prod}
                Let $\mathcal{S}$ be a stabilizer group of a maximal
                information state and $A$, $B$ two parties that hold
                each $N_A$ and $N_B$ such that $N_A + N_B = N$
                different elementary subsystems respectively, then the
                state is entangled if and only if $\mathcal{S} \neq
                \mathcal{S}_A \cdot \mathcal{S}_B$.
            \end{restatable} 
            
 \paragraph{Toy mixed state entanglement.}      
           For non-maximal information states the situation is more
           complicated because the state may not be factorizable, but can still be written as the sum of two product states and, therefore, not be entangled. In those cases, the stabilizer group may not factorize.
           Consequently, the above theorem holds for mixed states only
           in one direction:
           
           \begin{restatable}[Characterizing mixed entangled states]{theorem}{toystabprodtwo}\label{toy:stab:prod2}
            Let $\mathcal{S}$ be a stabilizer group of a non-maximal
            information state and $A$, $B$ two parties that hold each
            $k$ and $n-k$ elementary subsystems respectively. If $\mathcal{S} =
            \mathcal{S}_A \cdot \mathcal{S}_B$ then the state is not entangled.  The converse is not necessarily true. 
           \end{restatable}

     \paragraph{Examples.}      In the stabilizer formalism we can also describe the
           epistemic ``cat states'' defined in \cref{toy:org:ent_states} and,
           thus, show that they are valid. Let us define the
           injections $\pi_j,\pi'_j$
           in the following way  
           \begin{align}
            \pi(0) &= \pi_j(0) = 1 & \pi(1) &= \pi_j(1) = 2 \\
            \pi'(0) &= \pi'_j(0) = 3 & \pi'(1) &= \pi'_j(1) = 4.
           \end{align}
           In this case, each state in the ontic basis is stabilized by
           $span\{\mathcal{Z}_i \mathcal{Z}_{i+1}| \ i \in
           \{1,...,N-1\}\}$, because for each ontic state in the ontic
           basis all subsystems have their ontic state either in the $+1$
           eigenspace of $\mathcal{Z}$ or in the $-1$ eigenspace of
           $\mathcal{Z}$. This state is additionally stabilized by $
           \prod_{i = 1}^{N} \mathcal{X}_i$, because the ontic state of
           the last system in each state in the ontic basis ensures that
           there is always an even amount of systems that are in the
           eigenstate of $\mathcal{X}$. Thus, the epistemic state in
           \cref{toy:org:ent_states} with our particular choice of
           injections is stabilized by $\prod_{i = 1}^{N}
           \mathcal{X}_i$ and all elements in $ \{\mathcal{Z}_i
           \mathcal{Z}_{i+1}|\ i \in \{1,...,N-1\}\}$, which commute and
           are independent. There is no additional independent stabilizers as we already identified $N$ independent commuting
           stabilizers. Note that any other choice of injections corresponds to a local permutation of the initial choice of injectors -- this provides us with a procedure of finding the stabilizers for an arbitrary choice of injectors $\pi_j, \pi_j'$. 

\subsection{Generalization for an arbitrary dimension}
\label{subsec:entanglement-gen}

To define entanglement for a general case, we first need to understand what it
        means for a state to be completely uncorrelated. Consider two parties, Alice and Bob,
        where Alice has access to a subsystem $A$, and Bob -- to a subsystem $B$. 
        
        \begin{definition}[Product states (generalized)]
            Let $(V,\vec{v})$ be a state over two subsystems $A$ and $B$. We say that it is a product state if 
            $V =
        V_A \oplus V_B$, where $V_A$ is spanned by vectors
        which only have non zero entries on  $A$, and analogously for $V_B$. 
        \end{definition}
        
        In this case Alice, who has access to observables in $V_A$, knows nothing about the ontic
        state of system $B$,  and vice-versa. This definition does
        not depend on $\vec{v}$ at all, in the same way as entanglement
        in the stabilizer formalism did not depend on the sign of the
        stabilizers (that is, their valuation). We now proceed to defining entanglement and correlations.
 
         \begin{definition}[Generalized entanglement]\label{def:gen:ent}
                    We call an epistemic state $(V,\vec{v})$, shared by $A$ and $B$, \emph{correlated} if it is not product, and \emph{entangled} if cannot be
                    written as a mixture of product states.
                \end{definition}
        Let us illustrate these definitions for the generalized toy theory with few simple examples for $d=2$, which we accompany by their original toy theory counterparts.
        
        \begin{itemize}
        
            \item Entangled state: 
            \begin{equation}
                \left(\left\langle \begin{pmatrix} 1 \\[6pt] 0 \\[6pt]
            1 \\[6pt] 0\end{pmatrix},\begin{pmatrix} 0 \\[6pt] 1
            \\[6pt] 0 \\[6pt]
            1 \end{pmatrix} \right\rangle, \begin{pmatrix} 1 \\[6pt]  1
            \\[6pt]   1 \\[6pt] 
            1 \end{pmatrix} \right) \quad \Longleftrightarrow \quad \begin{tabular}{|{c}@{}|@{}{c}|{c}|{c}|}
                    \hline
                \colorcell{cyan}&\colorcell{white}&\colorcell{white}&\colorcell{white}\\ \hline
                \colorcell{white}&\colorcell{cyan}&\colorcell{white}&\colorcell{white}\\ \hline
                \colorcell{white}&\colorcell{white}&\colorcell{cyan}&\colorcell{white}\\ \hline
                \colorcell{white}&\colorcell{white}&\colorcell{white}&\colorcell{cyan}\\ \hline
                \end{tabular}
            \end{equation}
            
            \item Correlated state:
            \begin{equation*}
                \left(\left\langle \begin{pmatrix} 1 \\[6pt] 0 \\[6pt]
                    1\\[6pt] 0 \\[6pt]\end{pmatrix}\right\rangle,\begin{pmatrix} 1 \\[6pt]  0
                        \\[6pt]   0 \\[6pt] 
                        0 \end{pmatrix}\right)  = \left(\left\langle \begin{pmatrix} 1 \\[6pt] 0 \\[6pt]
                            1 \\[6pt] 0\end{pmatrix},\begin{pmatrix} 0 \\[6pt] 0
                            \\[6pt] 1 \\[6pt]
                            0 \end{pmatrix} \right\rangle, \begin{pmatrix} 1 \\[6pt]  0
                            \\[6pt]   0 \\[6pt] 
                            0 \end{pmatrix} \right) +  \left(\left\langle \begin{pmatrix} 1 \\[6pt] 0 \\[6pt]
                                1 \\[6pt] 0\end{pmatrix},\begin{pmatrix} 0 \\[6pt] 0
                                \\[6pt] 1 \\[6pt]
                                0 \end{pmatrix} \right\rangle, \begin{pmatrix} 0 \\[6pt]  0
                                \\[6pt]   1 \\[6pt] 
                                0 \end{pmatrix} \right) \end{equation*}
                                \begin{equation}
                                \Longleftrightarrow \quad \begin{tabular}{|{c}@{}|@{}{c}|{c}|{c}|}
                    \hline
                \colorcell{cyan}&\colorcell{white}&\colorcell{cyan}&\colorcell{white}\\ \hline
                \colorcell{white}&\colorcell{cyan}&\colorcell{white}&\colorcell{cyan}\\ \hline
                \colorcell{cyan}&\colorcell{white}&\colorcell{cyan}&\colorcell{white}\\ \hline
                \colorcell{white}&\colorcell{cyan}&\colorcell{white}&\colorcell{cyan}\\ \hline
                \end{tabular} = 
                 \begin{tabular}{|{c}@{}|@{}{c}|{c}|{c}|}
                    \hline
                \colorcell{white}&\colorcell{white}&\colorcell{white}&\colorcell{white}\\ \hline
                \colorcell{white}&\colorcell{cyan}&\colorcell{white}&\colorcell{cyan}\\ \hline
                \colorcell{white}&\colorcell{white}&\colorcell{white}&\colorcell{white}\\ \hline
                \colorcell{white}&\colorcell{cyan}&\colorcell{white}&\colorcell{cyan}\\ \hline
                \end{tabular} + 
                \begin{tabular}{|{c}@{}|@{}{c}|{c}|{c}|}
                    \hline
                \colorcell{cyan}&\colorcell{white}&\colorcell{cyan}&\colorcell{white}\\ \hline
                \colorcell{white}&\colorcell{white}&\colorcell{white}&\colorcell{white}\\ \hline
                \colorcell{cyan}&\colorcell{white}&\colorcell{cyan}&\colorcell{white}\\ \hline
                \colorcell{white}&\colorcell{white}&\colorcell{white}&\colorcell{white}\\ \hline
                \end{tabular}
            \end{equation}

            \item Product state:
            \begin{equation}
                \left(\left\langle \begin{pmatrix} 1 \\[6pt] 0 \\[6pt]
                                    1 \\[6pt] 0\end{pmatrix},\begin{pmatrix} 0 \\[6pt] 0
                                    \\[6pt] 1 \\[6pt]
                                    0 \end{pmatrix} \right\rangle, \begin{pmatrix} 1 \\[6pt]  0
                                    \\[6pt]   0 \\[6pt] 
                                    0 \end{pmatrix} \right) = \left\langle \begin{pmatrix} 1 \\[6pt] 0 \\[6pt]
                                            0 \\[6pt]
                                            0\end{pmatrix}\right\rangle
                                            \oplus \left\langle \begin{pmatrix} 0
                                            \\[6pt]
                                            0
                                            \\[6pt] 1 \\[6pt]
                                            0 \end{pmatrix} \right\rangle \quad \Longleftrightarrow \quad \begin{tabular}{|{c}@{}|@{}{c}|{c}|{c}|}
                    \hline
                \colorcell{cyan}&\colorcell{white}&\colorcell{cyan}&\colorcell{white}\\ \hline
                \colorcell{white}&\colorcell{white}&\colorcell{white}&\colorcell{white}\\ \hline
                \colorcell{cyan}&\colorcell{white}&\colorcell{cyan}&\colorcell{white}\\ \hline
                \colorcell{white}&\colorcell{white}&\colorcell{white}&\colorcell{white}\\ \hline
                \end{tabular} 
            \end{equation}
        \end{itemize}

\newpage
\section{Transformations}
\label{sec:transformations}

To characterize all possible transformations in  Spekkens' toy theory, we first have to consider what kind of transformations are allowed for ontic states. 

\subsection{Original toy theory}
\label{sec:org:trans}

\paragraph{Valid physical transformations.}
We call transformations $T$ that act on ontic states as \emph{physical transformations}. We denote their effect on epistemic states in the natural way: any epistemic state $E = \lor_i o_i $  is mapped to $T(E) = \lor_i T(o_i)$~\cite{Spekkens_2007}. 
We say a physical transformation is \emph{valid} if and only if it maps all valid epistemic states to valid epistemic states (also at the level of subsystems). 

\paragraph{Physical transformations must be bijective.} Suppose that \emph{many-to-one} transformations on ontic states are possible, for example $T: 1\mapsto 3, 2 \mapsto 3$. 
Then the epistemic state $(1 \lor 2)$ is mapped to $3$. Therefore, the output of this transformation violates the knowledge-balance principle and as such, many-to-one transformations are excluded. 
Now consider the case of \emph{one-to-many} transformations: for example, $T':1 \mapsto (3 \lor 4), 2 \mapsto 2$. Then the epistemic state $(1 \lor 2)$ would be mapped to $(2 \lor 3 \lor 4)$. This state is invalid, because it has an ontic basis of size 3. Hence, one-to-many transformations are also excluded, and the only allowed ontic state transformations are bijective, or one-to-one. In summary, physical transformations are reversible permutations of ontic states. However, being a permutation is not a sufficient condition for validity. 

\paragraph{Remark on irreversible transformations.}
There is still a way in the toy theory for irreversible transformations in which agents lose knowledge --- for example when their systems of interest interact with an environment they cannot track. 
Similarly to the Stinespring dilation in quantum mechanics, we can  formalize a transformation that does not preserve knowledge as a bijection of ontic states applied to a larger system, followed by tracing out a subsystem.  In other words, this is a physical, reversible transformation at the level of ontic states, followed by an epistemic ``action'' of forgetting a subsystem. These transformations do increase the size of the ontic basis of a state, and will be the subject of an upcoming work; we will also get back to them when we generalize to arbitrary dimensions. In the following we will focus only on what we call physical transformations, that is, permutations of ontic states.

\paragraph{Entangling transformations.} A transformation is \emph{local} if it can be decomposed into transformations acting only on individual subsystems. Beyond that, like in quantum theory, we are interested in characterizing the set of entangling transformations. 

\begin{definition}[Entangling transformations] A transformation is called \textbf{non-entangling} if and only if every product state on any division of subsystems $A_1$,...,$A_k$ is mapped to a product state of subsystems $A_1$,...,$A_k$; otherwise a transformation is called \textbf{entangling}.
\end{definition}

The following theorem  characterizes a class of local transformations.
\begin{restatable}[Sufficient condition for local transformations]{theorem}{toyorgtrafo}\label{toy:org:trafo}
    All non-entangling transformations that preserve the amount of information stored in each subsystem (that is, the size of the ontic support of each marginal doesn't change) are local.
\end{restatable}

\begin{restatable}[Non-entangling $\implies$  swaps and local permutations]{corollary}{toyorgperm}\label{toy:org:perm}
    All non-entangling transformations are composed of local permutations and system swaps. 
    A system swap $S_{i,j}$ is defined in the following way
    \begin{equation}
        S_{i,j}: o = a_1 \cdot a_2 \cdot ...
            \cdot a_i \cdot ... \cdot a_j \cdot ... \cdot a_n \to S_{i,j} o = a_1 \cdot a_2 \cdot ... 
            \cdot a_j \cdot ... \cdot a_i \cdot ... \cdot a_n.
    \end{equation}
\end{restatable}

                This corollary characterizes the non-entangling
                transformations in the original formulation of the toy theory. Entangling transformations, on the other hand, are
                easier to characterize in the stabilizer formalism.

\subsection{Stabilizer formalism}

          \paragraph{Conditions for allowed transformations.} In the stabilizer picture, any 
             physical
            transformation can be written as
            a permutation matrix $P$ that maps an ontic state
            $\vec{v}$ to another, $P \vec{v}$. If $\vec{v}$ is stabilized by
            $g$ then $P \vec{v}$ is stabilized by $PgP^{\dagger}$,
            because $P$ is unitary. To ensure that this is still a
            valid stabilizer we require that $PgP^{\dagger} \in G_N$.
            We want to ensure that stabilizer groups get mapped to
            stabilizer groups, therefore, we require that $PgP^{\dagger}$ and
            $PhP^{\dagger}$ commute if and only if $g$ and $h$
            commute. These two conditions fully characterize the
            subgroup of permutations that are valid \cite{Pusey_2012}. 
            In particular, a valid transformation $P$ is fully characterized by
            its action on a generating set of the toy Pauli group $G_N$.

            \paragraph{Toy CNOT gate.}
             The two-system permutation
                \renewcommand{\arraystretch}{0.7}
                \begin{tikzpicture}[overlay, remember picture, yshift=-1cm, shorten >=.5pt, shorten <=.5pt]
                    \draw [<->] ([yshift=2.5pt]{pic cs:a})  to ([yshift=2.5pt]{pic cs:b});
                    \draw [<->] ([yshift=2.5pt]{pic cs:c})  to ([yshift=2.5pt]{pic cs:d});
                    \draw [<->] ([yshift=3.5pt]{pic cs:f})  to ([yshift=3.5pt]{pic cs:i});
                    \draw [<->] ([yshift=1.5pt]{pic cs:e})  to ([yshift=1.5pt]{pic cs:j});
                    \draw [<->] ([yshift=2.5pt]{pic cs:g})  to ([yshift=2.5pt]{pic cs:l});
                    \draw [<->] ([yshift=2.5pt]{pic cs:h})  to ([yshift=2.5pt]{pic cs:k});
                \end{tikzpicture}
                    \begin{equation}
                        \begin{tabular}{|C|C|C|C|}
                            \hline
                            \tikzmark{f}& \tikzmark{g} & \tikzmark{i}&  \tikzmark{k} \\\hline
                            \tikzmark{e}& \tikzmark{h} &\tikzmark{j} &  \tikzmark{l} \\\hline
                             & \tikzmark{b} & &  \tikzmark{d} \\\hline
                             & \tikzmark{a} & &  \tikzmark{c} \\\hline
                          
                        \end{tabular}
                    \end{equation}
                    can be seen as an analog to the CNOT-gate.
            In
            the stabilizer formalism the CNOT gate can 
            be defined by the action on generators of $G_2$ \cite{Pusey_2012}:
            \begin{align}
                \begin{split}
                    \mathcal{X}_1 \to \mathcal{X}_1 \mathcal{X}_2, \ \mathcal{X}_2 \to \mathcal{X}_2\\
                    \mathcal{Z}_1 \to \mathcal{Z}_1, \ \mathcal{Z}_2 \to \mathcal{Z}_1\mathcal{Z}_2
                \end{split}.
            \end{align}

\paragraph{Clifford group.} The group of valid physical transformations in the toy theory is analogous to the quantum Clifford group. 
In quantum mechanics, the transformations of stabilizer states are elements of the Clifford group $C_N$, which is the group of  all unitary operations $U$ fulfilling
$ U^{\dagger}gU \in P_N, \quad \forall g \in P_N$, where $P_N$ is the Pauli group for $N$ qubits.
 These unitary operations transform stabilizer groups into stabilizer groups \cite{Gottesman_1998}. The same applies in the toy theory:
\begin{restatable}[\cite{Pusey_2012}]{lemma}{stabgroups}\label{stab:groups} 
    A stabilizer group $S = span\{g_1,...,g_k\}$ is transformed to a new stabilizer group $S' = span\{ U^{\dagger}g_1U,...,  U^{\dagger}g_kU\}$ if and only if the transformation $U$ belongs to the set
    $$ C_N = \{ U: U^{\dagger}gU \in G_N, \quad \forall g \in G_N\}, $$
    where $G_N$ is the toy Pauli group.
    Therefore $C_N$ forms the set of allowed unitary transformations in the toy theory. 
\end{restatable} 
Any transformation can be determined by its action on a
generating set of the toy Pauli group $G_n$, for example
$\{\mathcal X_k, \mathcal Z_k\}_{k=1,\dots,N}$. The quantum Clifford group is generated by the Hadamard gate, phase gate and the CNOT gate \cite{Gottesman_1998};           analogously,  all valid physical transformations in the toy theory are generated by the toy CNOT gate and single system transformations \cite{Pusey_2012}.

\subsection{Generalizing to arbitrary dimensions}
\label{toy:gen:trans}

\paragraph{Reversible physical transformations.}
        The underlying theory of ontic states is symplectic, therefore
        allowed transformations must also be symplectic, that is they
        should conserve the symplectic inner product. Because allowed
        transformations have to map epistemic states to epistemic states,
        the transformations must be affine
        \begin{equation}
            \vec{m} \to S \vec{m} + \vec{a}
        \end{equation}
where $\vec{a} \in \Omega$ and $S$ is a symplectic matrix, that is $S^T J S = J$, where $J$ is given by~\cref{eq:j}. Thus, $S$ conserves symplectic inner products. Importantly, $S$ also has an inverse if $d$ is not prime, because $S^{-1} = J^T S^T J$, and does not require for the numbers $a\in\mathbb{Z}_d$ to have an inverse (which does not necessarily exist if $d$ is not prime). These transformations are the equivalent of unitary transformations in quantum theory~\cite{SpekkensFoundations2016}.

        The probability distribution $\mu$ corresponding to the
        epistemic state $(V, \vec{v})$ is transformed as \cite{SpekkensFoundations2016}
        \begin{equation}
            \mu(\vec{m}) \to \mu'(\vec{m}) = \mu(S^{-1} \vec{m} - \vec{a}).
        \end{equation}
        From this definition we can derive the transformation
        rule of $(V,\vec{v})$, and determine the structure of the transformed vector space.
        
        \begin{restatable}[{Reversible transformations (generalized)}]{lemma}{toygentrafo}\label{toy:gen:trafo}
            Valid reversible transformations are \emph{sympletic transformations} represented by a pair $(S, \vec a)$, where $a \in \Omega$ and $S$ is a sympletic matrix. 
             An epistemic state $(V,\vec{v})$ transforms under such a symplectic
            transformation as
            \begin{equation}
                (V,\vec{v}) \to (\ (S^T)^{-1} V,\  S(\vec{v}+\vec{a})\ ).
            \end{equation}
        \end{restatable}
       
        \begin{restatable}[Validity of transformed  states]{lemma}{toygenspace}\label{toy:gen:space}    
            The transformed vector space defined in
            \cref{toy:gen:trafo} is isotropic.
        \end{restatable}

        \paragraph{Irreversible transformations.}
        Similarly to CPTP maps in quantum mechanics there are maps
        where the information about the system is not preserved. These
        maps are defined in analogy to the Stinespring dilation for
        CPTP maps. An irreversible map can be seen as a reversible map
        applied to the system and an ancilla, resulting in the probability
        distribution $\mu'(\vec{m}_s,\vec{m}_a )$, where $\vec{m}_s$ is
        the ontic state of the system, and $\vec{m}_a $ the ontic state
        of the ancilla. We then marginalize over the ancilla system
        $\mu'(\vec{m}_s) = \sum_{\vec{m}_a}
        \mu'(\vec{m}_s,\vec{m}_a )$~\cite{SpekkensFoundations2016}.

\newpage
\section{Discussion}
\label{sec:conclusions}

 \subsection{Relation to other theories}

\paragraph{Mimicking quantum behaviour.} We saw that the toy theory is equivalent to the stabilizer sub-theory of quantum mechanics. 
The toy theory reproduces  many quantum  phenomena~\cite{Spekkens_2007,SpekkensFoundations2016,Catani_2017,Disilvestro_2017}
including  interference, non-commutativity of observables, coherent superposition, collapse of the wave function, complementary bases, no-cloning~\cite{Wootters_1982}, teleportation~\cite{Bennett_1993}, key distribution~\cite{Bennett_2014}, locally indistinguishable product bases~\cite{Bennett_1999}, unextendible product bases~\cite{Bennett_1999_2}, some aspects of entanglement, the Choi-Jamio{\l}kowski isomorphism between states and operations~\cite{Choi_1975,Jamio_kowski_1972}, the Stinespring dilation of irreversible operations~\cite{Stinespring_1955}, error
   correction  and $k$-threshold key sharing~\cite{Disilvestro_2017}. 
However, there are quantum features that cannot be reproduced in the toy theory, such as Bell inequality violations~\cite{Bell_1964} and contextuality~\cite{Kochen,Spekkens_2005,Liang_2011}. This is because the toy theory is a local hidden variable theory, and its underlying ontological theory is non-contextual.\footnote{It has been possible to extend Spekkens' toy model to produce a Bell inequality violation; however, this extension is no longer epistemically restricted~\cite{van_Enk_2007}.}

\paragraph{Quadrature quantum mechanics.} In some cases Spekkens' toy theory is equivalent to quadrature quantum mechanics, a subtheory of quantum theory~\cite{SpekkensFoundations2016,Catani_2017}. The starting point of quadrature quantum mechanics are the
    projectors onto the position and momentum observables:
    \begin{equation}
        \Pi_{q}(\texttt{q})= \proj{ \texttt{q}}{ \texttt{q}}, \qquad
        \Pi_{p}(\texttt{p})= \proj{\texttt{p}}{ \texttt{p}},
    \end{equation}
    where we are considering a single (discrete or continuous) quantum system, $\{\ket{\texttt q}\}_{\texttt q }$ is an arbitrary basis which we name ``position basis'', and the momentum basis $\{\ket{\texttt p}\}_{\texttt p }$ is a conjugate basis, satisfying $\inprod{\texttt q}{\texttt p}= \frac 1{2 \pi \hbar} e^{i \frac{\texttt q\texttt p}{\hbar}}$ in the continuous case and $\inprod{\texttt q}{\texttt p}= \frac 1{\sqrt{d}} e^{i 2 \pi \frac{\texttt q\texttt p}{d}}$ in the discrete case, where $d$ is the dimension of the system. 
    To form the quadrature subtheory, we restrict ourselves to the subset of states of our Hilbert space for each we either know the position or momentum, that is, the states that can be represented by one of these projectors. 
    In other words, the valid states in this subtheory are those analogous
    to states in Spekkens' toy theory where one knows the value of the
    variables $q \sim \texttt{q}$ or $p \sim \texttt{p}$ respectively.
    If we have several such quantum systems, the valid states of the subtheory are those given by the tensor product of some of these projectors. It can be
    shown that for any two sets of commuting variables, there exists a
    symplectic transformation $S$ which transforms one set into the
    other. Therefore, with the the unitary
    representation of the symplectic group, we can apply this to the
    above projectors and find states analogous to any epistemic state
    $(V,\vec{v})$.
    The allowed transformations are given by the unitary
    representation of phase space displacements and the symplectic
    group; quadrature quantum mechanics is closed under these transformations. 
    All allowed measurements are sets of projectors which
    are obtained by applying a transformation to the projector valued
    measurement $\{\Pi_{q}(\texttt{q})\}_{\texttt{q}}$ \cite{SpekkensFoundations2016}.
    For even values of the dimension $d$ of individual systems, this connection cannot be made, since 
   then quadrature quantum mechanics becomes contextual, and thus cannot be
   equivalent to the non-contextual toy theory~\cite{Catani_2017}.

\paragraph{Liouville mechanics and Gaussian quantum mechanics.}
A classical example of a theory that distinguishes between ontic and epistemic states is Liouville mechanics~\cite{Liouville1838}. There, the epistemic state is characterized by a probability distribution over the phase space of ontic states. 
If we take Liouville mechanics and introduce an 
epistemic restriction similar to the Heisenberg uncertainty principle, we end up with a theory equivalent to Gaussian quantum mechanics~\cite{Bartlett_2012}.

\subsection{Summary of new results} 
\label{sec:discussion:results}

Besides the review, this paper introduced the following new technical contributions, which we hope bring Spekkens' toy theory into a more refined state for future work: 
\begin{enumerate}
    \item[Theorem~\ref{toy:org:number}] We translated the knowledge balance principle of the entire state into a statement about sets, and used it to find a necessary condition for epistemic states.

    \item[Corollary~\ref{toy:org:perm}] We proved in the original formulation that non-entangling transformations are decomposed into local transformations and system swaps.
    
    \item[Lemma~\ref{toy:gen:trafo}] We translated the definition of transformation: from an action on the probability distribution over phase space into an action on the epistemic state.

    \item[Section \ref{toy:gen:sup_mix}] We generalized the definition of superpositions and mixtures from the stabilizer formalism to the generalized framework. However, we were only able to prove soundness of the definition of superpositions in the case where the dimension is prime and the continuous case, and we leave the other cases for future work.
    
    \item[Definition~\ref{def:gen:ent}] We defined entanglement in the original formulation and its generalization to arbitrary dimensions. As far as we're aware, previously there was no definition of entanglement, but some epistemic states were known to be analogous to Bell states in quantum mechanics. We used this as a starting point for our definition of entanglement.
    
    \item[Lemma~\ref{toy:org:ent}] We introduced a criterion for entanglement in the original formulation, based on a similar statement for stabilizer quantum mechanics.
    
    \item[Theorem~\ref{toy:gen:measurementupdate}] We simplified the measurement update rule and unified the prime and non-prime dimensional cases. 
\end{enumerate}

\subsection{Open questions and next steps}

\paragraph{Observer and memories. } This review ends where our next paper begins: the groundwork has been laid to model observers' memories as physical systems, and each measurement as  the joint evolution of the system measured and the memory. As we will see, this is not without its peculiarities. For example, we will see that it is impossible for a physical agent in the toy theory to model arbitrary probability distributions (and mixtures of states). 

\paragraph{Dynamics.} Is there an equivalent of the Schr\"odinger equation, and Hamiltonian operators, for the toy theory? What constraints should they satisfy? As far as we know these are open questions which would help characterize the theory better. 

\paragraph{Subsystem decomposition.} In quantum mechanics, we can factor a global Hilbert space as the tensor product of several subsystems in uncountable ways: a state can be entangled in a factorization $\hilbert \sim \hilbert_A \otimes \hilbert_B$, and product in a different decomposition $\hilbert \sim \hilbert_C \otimes \hilbert_D$. In the toy theory, the subsystem structure is imposed rigidly, but it would be interesting to relax this assumption and allow for multiple subsystem decompositions of a a global system.   For example, consider an epistemic state like that of (\ref{eq:invalid_local_states}), which is valid if taken globally, but invalid at the level of our pre-determined subsystem decomposition (because the knowledge balance principle doesn't hold for every subsystem). However, there is a mapping to another decomposition of global ontic states into different subsystems for which the epistemic state is valid. This tells us that the validity of an epistemic state is relative to a choice of subsystem decomposition of the global space. It could be fruitful to study the implications of this, both in the toy theory and in quadrature quantum theory.

\paragraph{Where to learn more.} In this review we went over the formalism(s) of the toy theory. To explore applications like how to implement different information-processing tasks within the theory, we recommend the following resources: 
\begin{itemize}
    \item[\cite{Spekkens_2007}] The original article by Spekkens covers protocols for teleportation, dense coding, steering and entanglement swapping. It also includes proofs of the no-cloning theorem, lack of Bell inequality violations and non-contextuality of the theory. 
    
    \item[\cite{Pusey_2012}] Wherein Pusey introduces the stabilizer formalism and proves its equivalence to the original formalism for the toy theory. It also introduces the toy  Mermin-Peres square. 
    
    \item[\cite{SpekkensFoundations2016, Catani_2017}] Equivalence of the toy theory to stabilizer quantum theory (for odd prime dimensions in Spekkens' article \cite{SpekkensFoundations2016} and generalized to odd integers by  Catani and Browne \cite{Catani_2017}). Spekkens also proves that quadrature and stabilizer quantum mechanics are equivalent.
\end{itemize}

\vspace{0.7cm} 

\subsection*{Acknowledgements} 
We thank Rob Spekkens and Lorenzo Catani for valuable discussions.
NN and LdR acknowledge support from the Swiss
National Science Foundation through 
SNSF project No.\ $200020\_165843$ and through 
the National Centre of
Competence in Research \emph{Quantum Science and Technology}
(QSIT). 
LdR further acknowledges support from  the FQXi large grant \emph{Consciousness in the Physical World}. 

\subsection*{Author contributions}
This review is the first part of LH's semester project as a masters student. As such, the bibliographic research, novel technical contributions and first draft of the paper are hers. LdR and NN have supervised her and revised the manuscript. 


\newpage
\appendix

\addcontentsline{toc}{section}{\sc{Appendix}}

\section{Linear algebra lemmas}
\label{appendix:la}
Here we list the results from linear algebra we used in the refinement of the generalization of Spekkens' toy theory.

  \begin{lemma}[\cite{Fischer_2014}]\label{toy:gen:lin1}
            Let $W \subset \Omega$ be a subvector space or
            submodule and $\vec{w} \in \Omega$. Then for any $\vec{a} \in W+\vec{w}$

            \begin{equation}
                W + \vec{w} = W + \vec{a}
            \end{equation}
        \end{lemma}

        \begin{proof}
            Let $\vec{b} \in W +\vec{w}$ then
            \begin{equation}
                \vec{b} = \vec{w}_1 + \vec{w} 
            \end{equation}
            for some $\vec{w}_1 \in W$. As $\vec{a} \in W + \vec{w}$ we know that 
            \begin{align}
                \begin{split}
                    \vec{a} &= \vec{w_2} + \vec{w} \\
                    \iff \vec{w} &= \vec{a} - \vec{w}_2 
                \end{split}
            \end{align}
            for some $\vec{w}_2 \in W$. Plugging the expression for
            $\vec{w}$ into the expression for $\vec{b}$ we find:
            \begin{equation}
                \vec{b} = \vec{w}_1 - \vec{w}_2 + \vec{a} \in W + \vec{a}
            \end{equation}
            Therefore, $W +\vec{w} \subset W + \vec{a}$.

            Let $\vec{c} \in W + \vec{a}$ then we can write 
            \begin{equation}
                \vec{c} = \vec{w}_3 + \vec{a} = \vec{w}_3 + \vec{w}_2 +\vec{w} \in W + \vec{w}
            \end{equation}
            where we used the expression for $\vec{a}$ from above. From this
            equation we can conclude that $W +\vec{a} \subset W + \vec{w}$.
        \end{proof}

        \begin{lemma}[\cite{Gallier_2011}]\label{toy:gen:lin2}
            Let $W, V \subset \Omega$ be two subvector spaces or
            submodules and $\vec{v},\vec{w} \in \Omega$. Then if $(W +
            \vec{w}) \cap (V + \vec{v}) \neq \emptyset$, it holds that
            
            \begin{equation}
                (W + \vec{w}) \cap (V + \vec{v}) = (W \cap V) + \vec{u}
            \end{equation}
            with $\vec{u} \in (W + \vec{w}) \cap (V + \vec{v})$.
        \end{lemma}

        \begin{proof}
            If $(W + \vec{w}) \cap (V + \vec{v}) \neq \emptyset$ then there exists a
            $\vec{u} \in (W + \vec{w}) \cap (V + \vec{v})$.
            \Cref{toy:gen:lin1} allows us to write
            \begin{align}
                \begin{split}
                    W +\vec{w} &= W + \vec{u}\\
                    V + \vec{v} &= V + \vec{u}.
                \end{split}
            \end{align}
            This means each element in $V + \vec{v}$ is of the form
            $\vec{u}_1 = \vec{v}_1 + \vec{u}$ for $\vec{v}_1 \in V$,
            and each element in $W + \vec{w}$ is of the form
            $\vec{u}_2 = \vec{w}_1 + \vec{u}$ for $\vec{w}_1 \in W$.
            Therefore $\vec{u}_1$ is in $W + \vec{w}$ if and only if $\vec{v}_1
            \in W$ and $\vec{u}_2$ is in $V + \vec{v}$ if and only if $\vec{w}_1 \in V$.
            Therefore we can conclude that $(W + \vec{w}) \cap (V + \vec{v}) = (V \cap
            W)+ \vec{u}$. $ $
        \end{proof}

        \begin{lemma}[{\cite{Catani_2017}}]\label{toy:gen:lin3}
            Let $V,W \subset \Omega$ be two subvector spaces or
            submodules then it holds that 
            \begin{equation}
                (V \oplus W)^{\perp} = V^{\perp} \cap W^{\perp}
            \end{equation}
        \end{lemma}

        \begin{proof}
            Let $\vec{a} \in (V \oplus W)^{\perp}$ and, therefore, for all vectors
            $\vec{u} \in (V \oplus W)$ it holds that $\vec{a}^{T} \vec{u} = 0$. In
            particular, this holds for $\vec{a} \in V^{\perp}$ and $\vec{a} \in
            W^{\perp}$ as $V$ and $W$ are subsets of $(V \oplus W)$. Therefore, we can conclude that
            $(V \oplus W)^{\perp} \subset V^{\perp} \cap W^{\perp}$.

            Let $\vec{b} \in V^{\perp} \cap W^{\perp}$ and let $\vec{u} \in (V
            \oplus W)$ be arbitrary. Then we find that
            \begin{equation}
                \vec{u}^T \vec{b} = \vec{u}_w^T\vec{b} + \vec{u}_v^T \vec{b} = 0
            \end{equation}
            for $\vec{u}_w \in W$ and $\vec{u}_v \in V$ such that $\vec{u}_w
            + \vec{u}_v = \vec{u}$. Therefore, $V^{\perp} \cap W^{\perp} \subset
            (V \oplus W)^{\perp}$
        \end{proof}

        \begin{lemma}[\cite{Wilding_2013,Lam_2004}]\label{toy:gen:lin4}
            Let $V \subset \Omega$ be a subset or submodule. Then it
            holds that $(V^{\perp})^{\perp} = V$.
        \end{lemma}

        \begin{proof}
            The proof for this in the case for general $d$ can be found in
            \cite{Wilding_2013}. In the case for general vector spaces 
            ($d$ prime or the continuous case) the proof can be found
            in \cite{Lam_2004}.
        \end{proof}

\section{Proofs of new results}
\label{appendix:toyproofs}
\subsection{States}

Here you can find the proofs of statements listed in \cref{sec:states}.
\toyorgcanonical*
\begin{proof}
                ``$\implies$'' Let $\mathcal{Q}$ be a canonical set of
                questions, then it fulfills that each ontic states is
                uniquely defined through a set of answers. Thus, the
                intersection $\cap_{i} M_{i,a_i}$ will contain exactly
                one element.

                ``$\Longleftarrow$'' Let $\mathcal{Q}$ be a set of
                $2N$ questions such that the correspond set of partitions
                $(M_{i,0},O\setminus M_{i,0})= (M_{i,0},M_{i,1}) \in
                \mathcal{M}(\mathcal{Q})$ that fulfills $|\cap_{i} M_{i,a_i}| =
                1$ for all $2N$-bit strings $a$. No two bit strings
                $a,a'$ generate the same element in the intersection,
                otherwise this element would have to be in a set and
                its complement at the same time. Therefore there are
                $2^{2N}$ different states in the intersection which by
                the pigeon hole principle have to correspond to all
                different ontic states. Thus for each ontic state
                there is a set of answers which characterizes the
                ontic state uniquely.
                $ $
            \end{proof}
            
\toyorgbasis*
\begin{proof}
               For each answer in $a_i \in \mathcal{A}_{\mathcal E}$ we can conclude
               that the state is in $M_{i, a_i}$ and not in $O\setminus
               M_{i,a_i}$. Therefore all compatible states need to be
               contained in the intersection $\cap_{i \in \mathcal E}
               M_{i, a_i}$. 
               Conversely, all states in $\cap_{i \in \mathcal E}
               M_{i, a_i}$ are compatible with all the answers in $\mathcal A_{\mathcal E}$ by definition.
\end{proof}
\toyorgnumber*
\begin{proof}
                Let $\mathcal{P} = \{\mathcal{M}(Q_i)| a_i \notin
               \mathcal{A}_\mathcal{E} \}$
               the set partitions corresponding to questions in
               $\mathcal{Q}$ where the answer an is unknown. Then it
               must hold for each $2N-k$ bit string $\mathbf{b}$ and
               $\mathcal{M}(Q_i) \in \mathcal{P}$, $i \in \{0,...,2N-k\}$ 
               
                \begin{equation}
                    |E \cap (M_{0,b_0} \cap ... \cap M_{2N-k,b_{2N-k}})| = 1.
                \end{equation}
               The single element in the intersection is different for
               each bit string. Therefore, there are $2^{2N-k}$
               different elements in $E$ that correspond to different
               bit strings $\mathbf{b}$. Let there be an ontic state in $E$
               that does not correspond to a bit string $\mathbf{b}$. Then this
               ontic state needs to have different answers to the
               questions known by the observer, as it is characterized
               by a complete set of answers and the bit strings $\mathbf{b}$
               correspond to all different sets of complete answers
               compatible with the knowledge of the observer.
               Therefore all ontic states in $E$ are characterized by
               a bit string $\mathbf{b}$ and the cardinality of $E$ is
               $2^{2N-k}$. $ $
           \end{proof}
      \lemmaStateCorrespondenceGentab*
        \begin{proof}
            Let $(V,\vec v)$ be a general state. Then we can construct the corresponding stabilizer state. For each $\vec v \in V$ the corresponding stabilizer $g = \alpha p_1 \otimes ... \otimes p_n$ has the $p_j$ defined by 
            \begin{align}
                f_{2j-1} = 0, f_{2i} = 0 \implies p_j = \mathbb{1}_4 \\
                f_{2j-1} = 1, f_{2i} = 0 \implies p_j = \mathcal{X} \\
                f_{2j-1} = 0, f_{2i} = 1 \implies p_j = \mathcal{Z} \\
                f_{2j-1} = 1, f_{2i} = 1 \implies p_j = \mathcal{Y}.
            \end{align}
            and $\alpha = -2 \vec f^T v + 1$. 
            Now we need to show that this is a valid stabilizer state.
            We observe that for $n = 1$, the observables corresponding to $g$ and $g'$ commute if and only if the observables commute. 
            For $n \geq 1$, two observables $\vec f, \vec g$ with corresponding stabilizers $f,g$ respectively commute if and only if an even for an even amount of $i$s $f_{2i-1} g_{2i} - f_{2i} g_{2i-1} = 1$, which happens if and only if $p^f_i$ and $p^g_i$ the of the corresponding stabilizers $f,g$ respectively do not commute. As $f,g$ commute if and only if an even amount of $p^f_i$ and $p^g_i$ do not commute, $f,g$ commute if and only if $\vec f, \vec g$ commute. Therefore, the corresponding stabilizer state to $(V, \vec v)$ is valid. 
            The argument the other way around is analogous, except that we need to show that a valuation vector $\vec v$ exists. This is the case because $d = 2$ is prime and each observable attains all valuations. 
            
            The second statement can be seen in the following way. It can be easily checked for the correspondence between pairs of components $f_{2j-1} = 0, f_{2i} = 0$ and the $p_j$ the component wise addition of pairs is multiplication of the $p_j$. However, the multiplication of two stabilizers is just the multiplication of $p_j$ and the statement follows. 
        \end{proof}

\subsection{Measurements}
\label{appendix:toyproofs:measurements}
Here you can find the proofs of statements listed in \cref{sec:measurements}.        
       
       \toygenmeasurementupdate*
       \begin{proof}
        Similarly to quantum
        mechanics, we impose that repeated measurements of the same
        quantity should give the same result.
        One probability
        distribution that fulfills this property is
        \begin{equation}
            \mu_{(V',\vec{v}')}'(\vec{m}) \propto \delta_{V_{\pi}^{\perp} +\vec{v}_{\pi}}(\vec{m}) \mu_{(V,\vec{v})}(\vec{m}) = \frac{1}{N_V}\ \delta_{V_{\pi}^{\perp} +\vec{v}_{\pi}}(\vec{m})\ \delta_{V^{\perp} +\vec{v}}(\vec{m}),
        \end{equation}
        where  $N_V$ is the  normalization constant, as in \cref{eq:gen:prob_state}.  With this updated probability distribution, we can define the
        update rule for $V$. 
        
        First we consider the case where $V$ commutes with all variables in $V_\pi$.
        To find the updated state $(V',\vec{v}')$,
        we want to find all $\vec{m} \in \Omega$ such that
        $\mu_{(V',\vec{v}')}'(\vec{m})\neq 0$. This condition is
        fulfilled if and only if $\vec{m} \in V_{\pi}^{\perp} +
        \vec{v}_{\pi}$ and $\vec{m} \in V^{\perp} + \vec{v}$, that is,
        \begin{equation}
            V'^{\perp} + \vec{v}' = (V_{\pi}^{\perp} +  \vec{v}_{\pi}) \cap (V^{\perp} + \vec{v}).
        \end{equation}
        Notice that $V'^{\perp} + \vec{v}'$ is non empty, as otherwise the
        outcome $\vec{v}_{\pi}$ would appear with zero probability. 

        With \cref{toy:gen:lin1,toy:gen:lin2,toy:gen:lin3} we find that 
        \begin{align}
            \begin{split}
                V'^{\perp} + \vec{v}' &= (V_{\pi}^{\perp} \cap V^{\perp}) + \vec{r} \\ 
                         &= (V_{\pi} \oplus V)^{\perp} + \vec{r}
            \end{split}
        \end{align}
        with $\vec{r}$ such that $\vec{r} \in (V_{\pi}^{\perp} +  \vec{v}_{\pi}) \cap (V^{\perp} +
        \vec{v})$. In fact, we can choose any such $\vec{r}$, as all vectors in $(V_{\pi}^{\perp} +
        \vec{v}_{\pi}) \cap (V^{\perp} + \vec{v})$ lead to an equivalent
        measurement update. Using \cref{toy:gen:lin4}, we find the
        following measurement update rules: 
         \begin{equation}
                V' = V_{\pi} \oplus V,
            \end{equation}
            \begin{equation} 
                \vec{v}' \in (V_{\pi}^{\perp} +  \vec{v}_{\pi}) \cap (V^{\perp} +
                \vec{v}).
            \end{equation}
        As all vectors in $V$ commute with those in $V_{\pi}$ the $ V' = V_{\pi} \oplus V$ is isotropic. 
        
        In the case where some the vectors in $V_{\pi}$ do not commute
        with vectors in $V$ we have to modify the measurement update, because the rule of \cref{toy:gen:comm_update} produces a non-isotropic subspace. Like in the commuting case, we want that a
        repeated measurement produces the same outcome. Therefore, to
        satisfy the principle of classical complementarity,
        information about the pre-measurement state needs to get lost
        during a measurement. To get a valid post-measurement state we
        must remove all the variables $\vec{f}$ from $V$ that do not commute.
        Let $V_{\text{commute}}$ be the subspace of $V$ where all
        variables commute with all variables $V_{\pi}$. Then we can
        follow the same reasoning as in the commuting case with
        $V_{\text{commute}}$ instead of $V$ and find the final result. As $V_{commute}$ and $V_{\pi}$ commute by definition of $V_{commute}$ the updated state
        \begin{equation}\label{toy:gen:comm_update}
                V' = V_{\pi} \oplus V_{commute},
            \end{equation}
            \begin{equation} 
                \vec{v}' \in (V_{\pi}^{\perp} +  \vec{v}_{\pi}) \cap (V_{commute}^{\perp} +
                \vec{v}).
            \end{equation}
        is a valid state, by the same reasoning as above. 
       \end{proof}
    
 \toygenmeasurementupdateequivalence*
  \begin{proof}
  The measurement $V_{\pi}$ and a set of outcome vectors $\vec v_{\pi}^1,..., \vec v_{\pi}^k$ where each $ \vec v_{\pi}^j$ has a different valuation of the vectors in $V_{\pi}$. The updated state is a valid state, if we define it via the correspondence to the stabilizer formalism. By construction, the updated state is such that its ontic support is in the right part of the partition. This means we only need to show that it has maximal fidelity. Let us first calculate the fidelity of two arbitrary states $(V,\vec v),(V', \vec v')$. Either, $|(V^{\perp} + \vec v) \cap (V'^{\perp}+ \vec v')| = 0$ and the fidelity is zero, or the fidelity is non zero and $|(V^{\perp} + \vec v) \cap (V'^{\perp}+ \vec v')| = |(V\oplus V')^{\perp}| = 2^{2n - \dim (V\oplus V')}$. In the latter case, the fidelity is given by $2^{2n - \dim (V\oplus V')} \sqrt{2^{-2n + \dim V}2^{-2n + \dim V'}}$.
  Let us now apply this to the case where $(V, \vec v)$ is the state we do a measurement on, and $(V', \vec v')$ is the updated state. Let us define $d = \dim W$, where $W = V\cap V'$. Therefore, $\dim V' = d + d_{V' / W}$ and $\dim V = d + d_{V / W}$. We can now calculate the fidelity of these two states. If the fidelity between the two states is zero, then $(V',v')$ it is not a valid measurement update. Therefore the fidelity between the initial and updated state is given by 
  \begin{align}
      \begin{split}
          F((V, \vec v),(V', \vec v')) = 2^{ - (d + d_{V / W} + d_{V' / W})+(d+d_{V / W})/2 +(d + d_{V' / W})/2}
          = 2^{ - ( d_{V / W} + d_{V' / W})/2}.
      \end{split}
  \end{align}
    Therefore, the fidelity is the largest if $d_{V / W}$ $d_{V' / W}$ are the smallest. This means that the minimal amount of vectors are removed from or added to $V$. For $V'$ to be a valid measurement update it must hold that $V_{\pi} \subseteq V'$, therefore we need to add all vectors from $V_{\pi}$ to $V$ which are not already in $V$ and we need to remove all vectors from $V$ which do not commute with $V_{\pi}$. However, this is just the vector space $V_{\pi} \oplus V_{commute}$, which is $V'$ in our definition of the measurement update.
  \end{proof}

\subsection{Mixtures and superpositions}

Here you can find the proofs of statements listed in \cref{sec:mixtures}.
\toystabmix*
\begin{proof}
                Let $h_{1},...,h_k$ be independent stabilizer such
                that together with $g_1, ...,g_{N-k}$ they span a
                stabilizer group $\mathcal{T} = span\{g_1, ...,g_{N-k},
                h_{1},...,h_k\}$ of a maximal information state. With
                the same argument as in the stabilizer formalism one
                can show that it is in fact possible to find such
                $h_{1},...,h_k$.

                We define $2^{k}$ rephasings of $T$
                
                \begin{equation}
                    \mathcal{T}_{i^{(0)}}^{(0)} = span\{g_1, ...,g_{N-k}, (-1)^{i^{(0)}_1}
                h_{1},..., (-1)^{i^{(0)}_k} h_k\},
                \end{equation}
                where $i^{(0)}_j$ is the $j$-th digit of the binary
                representation of $i^{(0)}$. Then we can pair up all
                $\mathcal{T}^{(0)}_{i^{(0)}}$ such that the pairs have
                the same sign in front of all generators except the
                last generator. Then we take the mixture of each of the pairs
                and define these as $\mathcal{T}^{(1)}_{i^{(1)}}$
                where $i^{(1)}$ is obtained from $i^{(0)}$ by setting
                $i^{(0)}_k$ to zero in the binary representation.
                There are now only $2^{k-1}$
                $\mathcal{T}^{(1)}_{i^{(1)}}$ left. Repeating this
                procedure reduces the number of $h$ by one and halfs
                the number of $\mathcal{T}^{(1)}_{i^{(1)}}$. Thus, if
                this procedure is applied $k$ times there is only one
                $\mathcal{T}^{(k)}_{i^{(k)}}$ left, which is generated
                by $g_1, ...,g_{N-k}$.
            \end{proof}
\toygensuperpose*            
\begin{proof}
            For this definition to be sound, we need to prove that there
        always exists a $\vec{w} \notin V$ which commutes with
        $\vec{f}^{(1)}, ..., \vec{f}^{(k-1)}$. Let $U = \langle
        \vec{f}^{(1)}, ..., \vec{f}^{(k-1)} \rangle$ then $\vec{w}$
        has to be contained in $U^{\perp_{\text{symplectic}}} =
        \{\vec{m} \in \Omega| \ \forall \vec{v} \in V \ \vec{v}^T J
        \vec{m} = 0\}$. Because $U$ is isotropic it holds that $U
        \subset U^{\perp_{\text{symplectic}}}$. Let $\ell = dim(U^{\perp_{\text{symplectic}}})-dim(U)$. It
        holds that $\dim(U) + \dim(U^{\perp_{\text{symplectic}}}) =
        (k-1) + (k-1 +\ell) =2n$ \cite{Lam_2004}. Rearranging this
        expression, we can calculate $\ell = 2n-2k+2$. This implies
        that $\ell \geq 2$,
        because $k$ is upper bounded by $n$. Therefore, there always
        exists a vector $\vec{w} \in U \setminus
        U^{\perp_{\text{symplectic}}}$ which is linearly independent
        of $\vec{f}^{(k)}$ and, thus, not contained in $V$. 
        \end{proof}

\toygensuperposestabilizer*

    \begin{proof}
        Let $(V, \vec v_1)$ and $(V,\vec v_2)$ with $V = \langle \vec f_1, \dots, \vec f_k \rangle$ two states where $\vec f^k$ is completely unknown. Let $g_k$ be the stabilizer corresponding to $\vec f_k$ with valuation vector $\vec v$. Then $ g_1,...,g_k $ are independent due to the second point in \cref{toy:gen:stabcorrespondence}. The corresponding stabilizer state to $(V, \vec v_1)$ is then $\mathcal{S} = \langle g_1,...,g_k \rangle$. As, $f_k$ is totally unknown the corresponding stabilizer state of $(V,\vec v_2)$ is $\mathcal{S}' = \langle g_1,...,- g_k \rangle$. A superposition of $(V, \vec v_1)$ and $(V,\vec v_2)$ is $(V' = \langle \vec f_1, \dots, \vec f_{k-1} \rangle + \vec w, \vec v_1)$ wit $\vec w$ such that $V'$ isotropic. Then, let $h$ be the corresponding stabilizer to $\vec w$ with valuation $\vec v_1$. Then, $h$ is independent of and commutes with $g_1,..., g_{k-1}$. Which means that $h$ fulfils the definition of $h$ in the definition of superposition for stabilizers. 
        
        Let $\mathcal{S} = \langle g_1,...,g_k \rangle$ and $\mathcal{S}' = \langle g_1,...,- g_k \rangle$ be two stabilizer states. Then the corresponding general states 
        $(V, \vec v_1)$ and $(V,\vec v_2)$ with $V = \langle \vec f_1, \dots, \vec f_k \rangle$ are such that $\vec f^k$ is completely unknown.
        Let $h$ be a stabilizer as in the definition of the superposition of stabilizer states. Then the corresponding observable to $h$ fulfils the definition of $\vec w$ in the definition of superposition of general states. Let $\vec v'$ be the valuation vector corresponding to $\langle g_1,...,g_{k-1},h \rangle$. Then $\vec v'$ has the same valuation as $\vec v_1$ and $\vec v_2$ on $\vec f_{1},..., \vec f{k-1}$ and either the valuation of $\vec v_1$ or $\vec v_2$ on $\vec f_k$. Therefore, $\vec v'$ is an equivalent valuation vector to either $v_1$ or $v_2$ and $(\langle \vec f_{1},...,\vec f{k-1},\vec w \rangle, \vec v')$ is the superposition of $(V, \vec v_1)$ and $(V,\vec v_2)$. 
\end{proof}

\subsection{Entanglement}

Here you can find the proofs of statements listed in \cref{sec:entanglement}.

\toyproduct*
 \begin{proof}
                    See the $\Longleftarrow$ part of the proof of \cref{toy:stab:prod}.
                \end{proof}

\toystabprod*
\begin{proof}
            $\implies$ If $\mathcal{S} = \mathcal{S}_A \cdot
               \mathcal{S}_B$ then $\mathcal{S}=
               span\{s_{A,1},...,s_{A,k},s_{B,1},...,s_{B,N-k}\}$
               where $s_{A,i}$, $s_{B,i}$ are stabilizer operator only
               acting non trivially on the $A$ and $B$ system
               respectively. Therefore, there are $2^{2N_A-k}$ ontic
               states that are stabilized by $S_A$ on system $A$ and
               $2^{2N_B-N+k}$ that are stabilized by $S_B$ on system
               $B$. Taking the product of these states results in a
               state with an ontic basis of size
               $2^{2N_A-k}2^{2N_B-N+k} = 2^{2(N_A + N_B)-N} = 2^{N}$.
               The product state is stabilized by $\mathcal{S}$, as
               each element of the ontic basis of the product state is
               on the $A$ part stabilized by all $s_{A,i}$ and on the
               $B$ part stabilized by all $s_{B,i}$. As $\mathcal{S}$
               is a pure state we have found there are no more ontic
               states that can be stabilized by $\mathcal{S}$ and
               therefore the state stabilized by $\mathcal{S}$ is
               pure.

               $\Longleftarrow$ Let $\mathcal{S}$ be the stabilizer
               group corresponding to a product state, then we can
               find the stabilizer group $\mathcal{S}_A$ that
               stabilizes $A$ and the stabilizer group $\mathcal{S}_B$
               that stabilizes the state on the $B$ system. These
               stabilizers from a new group $\mathcal{S'} =
               \mathcal{S}_A \cdot \mathcal{S}_B$. We have found a
               different stabilizer group that stabilizes the same
               state as $\mathcal{S}$. If $\mathcal{S'} \neq
               \mathcal{S}$ then there must be a stabilizer in
               $\mathcal{S'}$ that is not in $\mathcal{S'}$.
               Therefore, there must be an ontic state stabilized by
               $\mathcal{S}$ but not by $\mathcal{S}$, but then the
               states stabilized are not the same. Thus, the two
               groups have to be the same and $\mathcal{S} =
               \mathcal{S}_A \cdot \mathcal{S}_B$.
            \end{proof}

\toystabprodtwo*           
\begin{proof}

               The same argument as in the $\Longleftarrow$ part of
               the previous theorem holds. The reason the
               $\Longrightarrow$ direction does not hold is because
               for mixed states non-entangled states are not
               necessarily product states.
\end{proof}

\subsection{Transformations}

Here you can find the proofs of statements listed in \cref{sec:transformations}.
\toyorgtrafo*
\begin{proof}
                    Local transformations are non-entangling and
                    subsystem information preserving, as it just
                    permutes the ontic states of each elementary
                    system.

                    Let $T$ be a non-entangling subsystem information
                    preserving transformation on $N$ systems, $O_N$
                    the set of ontic states on $N$ systems and $C = (a
                    \lor b) \cdot (1 \lor 2 \lor 3 \lor 4)^{N-1}$ a
                    valid epistemic state (i.e $a \neq b$). Then
                    \begin{align}
                        \begin{split}
                            T(c) &= \lor_{o \in O_{N-1}}T((a \lor b) \cdot
                            o) \\
                                &= (c \lor d)
                        \cdot (1 \lor 2 \lor 3 \lor 4)^{N-1}
                        \end{split}
                    \end{align} 
                     with $c \neq d $, because $T$ is a valid non-entangling
                    subsystem preserving transformation. Therefore, for any ontic
                    state $o$ it holds that
                    \begin{equation}
                        T((a \lor b) \cdot o) = (c \lor d) \cdot o'
                    \end{equation}
                    
                    $T$ is a permutation, therefore the amount of
                    states in the intersection of two ontic bases is
                    invariant under transformation. Furthermore, it
                    does not matter if the states are first
                    transformed and then the intersection is taken or
                    if first the intersection is taken and then the
                    states are transformed. Let $C' = (a\lor f) \cdot
                    (1 \lor 2 \lor 3 \lor 4)^{N-1}$ be an valid
                    epistemic state such that $a \neq f$ and $b \neq f$. Then
                    the same conclusions that hold for $C$ also hold
                    for $C'$. Therefore, for any ontic state $o$

                    \begin{equation}
                        T(a \cdot o) = T((a \lor b) \cdot o) \cap T((a \lor f) \cdot o) = g \cdot o'
                    \end{equation}
                    with $g$ either $c$ or $d$ and $o'$ some ontic
                    state. The same argument also applies for $b$.
                    Therefore, the transformation can be written as
                    \begin{equation}
                        T((a \lor b) \cdot o) =  T_1(a \lor b) \cdot
                    T_{N-1}(o).
                    \end{equation}
                    Repeating the same argument for $T_{N-1}(o)$ the
                    transformation $T$ factorizes into local
                    transformations.
                \end{proof}
\toyorgperm*                
\begin{proof}
                    Let $T$ be a non entangling transformation. Then
                    we can detect changes in the system by applying
                    transformation to $C = (a \lor b) \cdot (1 \lor 2
                    \lor 3 \lor 4)^{N-1}$ and detecting which system
                    $i$ is in the pure state after the transformation.
                    Only one system can be in the pure state, as the
                    transformation is a permutation of ontic states. 
                    
                    Consider the transformation $S_{1,i} T$, for which
                    the same argument as in the proof above holds.
                    Thus, this transformation can be decomposed in the
                    following way
                    \begin{equation}
                        S_{1,i} T = T_{1} \cdot T_{N-1}
                    \end{equation}
                    where $T_{1}$ a transformation on the first and $
                    T_{N-1}$ a transformation on the last $N-1$ subsystems.
                    Because $S_{1,i}$ is its own inverse, we can
                    rewrite rewrite the above decomposition
                    \begin{equation}
                        T = S_{1,i} (T_{1} \cdot T_{N-1}).
                    \end{equation}
                    Applying the same argument as above inductively to
                    $T_{N-1}$ factorizes the transformation in swaps
                    $S_{i,j}$ and local transformations. 
                \end{proof}
\stabgroups*            

\begin{proof}
                The elements of the set $\{ U^{\dagger}g_1U,...,
                U^{\dagger}g_kU\}$ still commute as 
                \begin{equation}
                    [U^{\dagger}g_iU,U^{\dagger}g_jU] = U^{\dagger}[g_i,g_j]U = 0.
                \end{equation}
                Therefore, $S'$ is a stabilizer group. 
            \end{proof}                
\toygentrafo*
\begin{proof}
            We can find the updated epistemic
        $(V',\vec{v}')$ state by calculating the set $V'^{\perp} +
        \vec{v}'$ which is characterized by vectors $\vec{m}$ such
        that $\mu'(\vec{m}) \neq 0$. For such vectors we can do the
        following chain of reasoning
        \begin{align}
            \begin{split}
                \mu'(\vec{m}) \neq 0  
            & \iff \mu(S^{-1} \vec{m} - \vec{a}) \neq 0 \\
            & \iff S^{-1} \vec{m} - \vec{a} \in V^{\perp} + \vec{v}
            \end{split}
        \end{align}
        Therefore, we can conclude that if and only if $\mu'(\vec{m})
        \neq 0 $ it holds that $\vec{m} = S(\vec{v}^{\perp} +\vec{v} +
        \vec{a})$ with $\vec{v}^{\perp} \in V^{\perp}$. This means we can
        infer $\vec{m} \in SV^{\perp} + S(\vec{v} + \vec{a})$ if
        $\mu'(\vec{m})
        \neq 0 $. On the other hand, if $\vec{m}
        \in SV^{\perp} + S(\vec{v} + \vec{a})$ then 
        \begin{align}
            \begin{split}
                \mu'(\vec{m}) &=
                \mu'(S(\vec{v}^{\perp} + \vec{v} + \vec{a})) \\
                &= \mu(S^{-1}(S(\vec{v}^{\perp} + \vec{a} + \vec{v}))-\vec{a})\\
                &= \mu(\vec{v}^{\perp} + \vec{v}) \neq 0.
            \end{split}
        \end{align}
        We can conclude that 
        \begin{equation}
            V'^{\perp} + \vec{v}' = SV^{\perp} + S(\vec{v} + \vec{a}).
        \end{equation}
        This means that the transformed valuation vector $\vec{v}'$ is given by
        \begin{equation}
            \vec{v}' = S(\vec{v} + \vec{a})
        \end{equation}
        and the updated vector space $V'$
        \begin{equation}
            V' = (SV^{\perp})^{\perp}.
        \end{equation}
        This expression can be evaluated using the definition of an
        orthogonal complement
        \begin{align}
            \begin{split}
                \vec{w} \in (SV^{\perp})^{\perp} &\iff \forall \vec{v} \in SV^{\perp}: \ \vec{v}^T \vec{w} = 0 \\
                &\iff \forall \vec{v} \in V^{\perp}: \ (S \vec{v})^T \vec{w} =0 \\
                &\iff  \forall \vec{v} \in V^{\perp}: \ \vec{v}^T (S^T \vec{w}) = 0 \\
                &\iff (S^T \vec{w}) \in V \\
                &\iff \vec{w} \in (S^T)^{-1} V.
            \end{split}
        \end{align}
        Therefore we can give a more direct expression for $V'$
        \begin{equation}
            V' = (S^T)^{-1} V.
        \end{equation}
        In summary a state $(V,\vec{v})$ transforms under a symplectic
        transformation $S \cdot + \vec{a}$ in the following way
        \begin{equation}
            (V,\vec v) \to ((S^T)^{-1} V, S(\vec{v}+\vec{a}))
        \end{equation} 
        \end{proof}                
\toygenspace*
\begin{proof}
            Let $V = \langle
            \vec{f}_1, ..., \vec{f}_k \rangle$ be an isotropic space. Then
            $V'$ is given by $V' = \langle (S^T)^{-1}\vec{f}_1,
            ...,(S^T)^{-1} \vec{f}_k \rangle$. The symplectic inner
            product between $(S^T)^{-1}\vec{f}_i$ and
            $(S^T)^{-1}\vec{f}_j$ can then be calculated
            \begin{align}
                \begin{split}
                ((S^T)^{-1}\vec{f}_i)^T J ((S^T)^{-1}\vec{f}_j )&= \vec{f}_i^T (S^{-1} J (S^T)^{-1}) \vec{f}_j \\
                &=  \vec{f}_i^T (J^T S^T J) J (J^T S J) \vec{f}_j\\
                &=  \vec{f}_i^T (J^T J J) \vec{f}_j\\
                &= \vec{f}_i^T J \vec{f}_j \\
                &= 0,
                \end{split}
            \end{align}
            were we used that $J^T = J^{-1}$ and the expression for
            $S^{-1}$.
        \end{proof}





\newpage
\bibliographystyle{unsrtnat}



\end{document}